\documentclass[12pt]{article}

\usepackage{amsfonts}
\usepackage{multicol}
\usepackage{amsthm}
\usepackage{graphics}
\usepackage{graphicx}
\usepackage{amsmath}
\usepackage{amsxtra}
\usepackage{amstext}
\usepackage{amssymb}
\usepackage{latexsym}
\DeclareFontFamily{U}{lasy}{}
\DeclareFontShape{U}{lasy}{m}{n}{<-> s * lasy10}{}
\DeclareFontShape{U}{lasy}{b}{n}{<-> s * lasyb10}{}
\usepackage[dvipsnames]{xcolor}
\usepackage{comment}

\usepackage{pgfplots}
\pgfplotsset{compat=1.18}

\usepackage{newpxtext,newpxmath}

\usepackage[authoryear]{natbib}
\usepackage[pdftex,bookmarks=true,bookmarksopen=true,pdfpagemode=UseOutlines,pdfstartview=FitH]{hyperref}

\usepackage{bbm}
\usepackage{dsfont}

\usepackage{chngcntr}
\usepackage{aliascnt}
\usepackage{apptools}
\AtAppendix{\counterwithin{lemma}{section}}
\AtAppendix{\counterwithin{proposition}{section}}

\usepackage{caption}
\usepackage{subcaption}
\usepackage{tikz}
\usepgflibrary{shadings}
\usetikzlibrary{shapes.geometric}
\usepackage[T1]{fontenc}
\definecolor{uncblue}{RGB}{75, 165, 211}
\definecolor{nberblue}{RGB}{0, 90, 155}
\definecolor{resgreen}{RGB}{34, 84, 67}
\definecolor{pennblue}{RGB}{0, 44, 119}
\definecolor{pennred}{RGB}{152, 30, 50}

\usepackage{setspace}

\usepackage{amsmath}

\DeclareMathOperator*{\argmin}{arg\,min}

\newtheorem{theorem}{Theorem}

\newtheorem{claim}{Claim}

\newtheorem{corollary}{Corollary}

\theoremstyle{definition}
\newtheorem{definition}{Definition}
\theoremstyle{plain}

\newtheorem{lemma}{Lemma}

\newtheorem{proposition}{Proposition}
\theoremstyle{remark}
\newtheorem{remark}[theorem]{Remark}
\theoremstyle{plain}

\hypersetup{%backref,
	pdfborder = {0 0 0},
	urlbordercolor = {0 0 0},
	colorlinks=true,
	linkcolor=pennblue,
	urlcolor=pennred,
citecolor=pennblue}

\usepackage{float}

\begin{document}
\newcommand{\reals}{\ensuremath{\mathbb{R}}}
\newcommand{\R}{\ensuremath{\mathbb{R}}}
\newcommand{\Z}{\ensuremath{\mathbb{Z}}}
\newcommand{\E}{\ensuremath{\mathbb{E}}}
\newcommand{\expect}{\ensuremath{\mathbb{E}}}
\newcommand{\dint}{\ensuremath{\,\mathrm{d}}}
\newcommand{\signalprofile}{signal profile\xspace}
\newcommand{\noise}{noisemaker\xspace}
\newcommand{\signal}{\ensuremath{s}}
\newcommand{\shirk}{\ensuremath{y}}
\newcommand{\signals}{\ensuremath{\mathbf{s}}}
\newcommand{\shirks}{\ensuremath{\mathbf{y}}}

\newcommand{\action}{\ensuremath{a}}
\newcommand{\actions}{\ensuremath{\mathbf{a}}}
\newcommand{\actionspace}{\ensuremath{\mathcal{A}}}
\newcommand{\permutation}{\ensuremath{\mathbf{i}}}
\newcommand{\contracts}{\ensuremath{\mathcal{W}}}
\newcommand{\contract}{\ensuremath{w}}

\newcommand{\state}{\ensuremath{x}}
\newcommand{\statedist}{\ensuremath{F}}

\newcommand{\FL}{\color{pennred} }
\newcommand{\DZ}{\color{nberblue} }
\newcommand{\SI}{\color{orange} }

\title{Monitoring Hierarchies in Knowledge Production
\thanks{Ichihashi: Queen's University (\href{mailto:s.ichihashi@queensu.ca}{\texttt{s.ichihashi@queensu.ca}}); Li: University of North Carolina at Chapel Hill (\href{mailto:lifei@email.unc.edu}{\texttt{lifei@email.unc.edu}}); Zou: University of North Carolina at Chapel Hill (\href{mailto:dihan@email.unc.edu}{\texttt{dihan@email.unc.edu}}). We have benefited from helpful discussions with Yu Awaya, Arjada Bardhi, Gary Biglaiser, Nina Bobkova, Aislinn Bohren, Doruk Cetemen, Yi-Chun Chen, Jaden Chen, Joyee Deb, Rahul Deb, Francesc Dilmé, Miaomiao Dong, Marina Halac, Wei He, Ashwin Kambhampati, Jacob Kohlhepp, Maciej H. Kotowski, Tan Gan, Jiangtao Li, Erik Madsen, Margaret Meyer, Peter Norman, Ichiro Obara, Paula Onuchic, Marcin Pęski, Ferdinand Pieroth, David Rahman, Joao Ramos, Anna Sanktjohanser, Satoru Takahashi, Can Tian, Curtis Taylor, Huseyin Yildirim, Kun Zhang, and audiences at numerous conferences and seminars. We are especially indebted to Steven Callander, whose constructive and detailed suggestions greatly improved the paper.}} 

\author{Shota Ichihashi \and Fei Li  \and Dihan Zou}

\date{\today}
\maketitle

\begin{abstract}
We study peer monitoring design in knowledge production. A principal leads agents who work on different but related tasks. By working on their own tasks, agents acquire information that is useful for evaluating their peers’ performance. We consider robust contracts under which effort by all agents is the unique rationalizable outcome. The optimal contract induces a hierarchy of monitoring authority: agents easier for the principal to monitor directly are assigned greater authority to evaluate other agents. Task diversity improves information value but weakens peer monitoring. This tension distorts optimal public task assignment toward the principal’s core expertise. Confidential random assignment relaxes this distortion by decoupling information production from the perceived monitoring structure.
	\thispagestyle{empty}

\end{abstract}

\clearpage

\pagenumbering{arabic} 

\section{Introduction}

Knowledge production accounts for a large share of economic activity in modern economies.\footnote{See \url{https://ncses.nsf.gov/pubs/nsb20247/executive-summary}.}  Much of this activity involves complex projects, which are divided into tasks and assigned to different agents. Providing incentives requires evaluating agents' performance, yet such evaluation is difficult because it depends on task-specific knowledge. At the same time, agents acquire expertise through their own work and are often well positioned to assess closely related tasks.
This makes peer monitoring a natural organizational response.

One example is software engineering. Engineers who work on related components of a system often understand the relevant code, dependencies, and failure modes better than outside evaluators. Code review is a common institutionalized way to bring such peer expertise into the assessment of technical work \citep{sadowski2018modern}. 
Another example is experimental knowledge production, such as materials science or drug development. Researchers who work under closely related experimental conditions are often better able to judge whether results are plausible, procedures were followed, or anomalies reflect mistakes rather than discoveries.

The use of peer monitoring for incentive provision raises several challenges. First, agents acquire different knowledge through their own work and hence differ in their ability to evaluate other tasks. How should \emph{monitoring authority} be structured across agents? Second, the information used for peer monitoring is itself produced through agents’ costly effort. Consequently, an agent’s incentive to work may depend on whether she expects her evaluators to work, whose incentives may in turn depend on still other agents. How can an organization establish \emph{credible monitoring} when agents’ incentives to acquire information are interdependent?  Third, the effectiveness of peer monitoring depends on how closely agents’ tasks are related. What implications does this have for the way organizations \emph{assign tasks}?

We propose a principal-agent model of spatial learning.
The principal is a lead scientist or engineer who manages a team of agents. Each agent is assigned to a task, represented as a location in a unit interval. The one-dimensional representation of the task space is a parsimonious way to capture a dimension of task relatedness. Each location has a stochastic local state, representing the uncertainty that knowledge production seeks to resolve.\footnote{In software development example, this dimension may be a user segment, a market, or a component architecture; in experimental research, it may be an experimental condition such as temperature, pressure, or compound composition.} An agent who exerts costly effort learns the state at her location and reports it to the principal; if she shirks, her report is uninformative. The principal observes the state at location zero, which represents her own core expertise within the project, but does not observe the states of agents’ tasks. Local states are generated by a Markov process, so nearby tasks are more likely to have the same state.
Thus, distance in the task space measures task relatedness.

Agents’ work creates no direct production externalities, but it generates information useful for evaluating related tasks. Leveraging this informational spillover, the principal designs the peer-monitoring structure through the contract: agent $j$ monitors agent $i$ if $j$’s report affects $i$’s wage. Peer monitoring, however, creates strategic interdependence. An agent may be reluctant to exert effort if she is unsure whether the peer assigned to evaluate her will also work and acquire the knowledge needed to assess her task. We therefore study robust contracts under which effort by all agents is the unique rationalizable outcome, equivalently, the unique outcome surviving iterated elimination of strictly dominated strategies (IESDS). This establishes monitoring credibility without relying on equilibrium coordination.

We first show that robust implementation induces a hierarchical allocation of monitoring authority. Reports from the principal and higher-layer agents may affect the wages of lower-layer agents, but not conversely. The hierarchy corresponds to the order of reasoning generated by IESDS. Agents whose shirking strategies are eliminated earlier must work regardless of how agents whose shirking strategies are eliminated later behave. They therefore occupy higher layers of the monitoring hierarchy. Once their incentives are secured, their reports become credible sources of information for disciplining later agents. Conversely, conditioning an earlier-eliminated agent’s wage on the report of a lower-layer, later-eliminated agent is never useful in a cheapest robust contract.

We then characterize the optimal monitoring hierarchy. Because the principal’s signal is more informative about tasks closer to her location, agents are ranked by direct monitorability from the principal. The optimal robust contract follows this order. The principal first disciplines the agent whose task is closest to her own expertise; that agent’s signal becomes credible and is used to discipline the next agent; and so on. The resulting optimal monitoring structure is a chain. In applications, this is a hierarchy of evaluative authority: a lead scientist or tech lead uses credible information from agents working near her core expertise to evaluate agents working on more distant or specialized tasks. Under this chain structure, the high-dimensional contract-design problem is reduced to a sequence of local incentive problems. The robust requirement makes the resulting incentive cost additive: each agent’s contribution depends only on the distance between her task and her immediate monitor's task.

Robust implementation creates a wedge between informational authority and monitoring authority. Along the target outcome, all agents work, so a lower-layer agent’s report may be informative about the performance of an agent above her in the hierachy. From the perspective of the sufficient-statistic principle, one would like to use such information in evaluation. Robustness prevents this because monitoring authority must flow downward in the hierarchy. The resulting underuse of information makes robust contracts more expensive than contracts that rely on upward or mutual monitoring, which can reduce payments but create circular incentive dependence and equilibrium multiplicity.

We next derive the implications of this hierarchy for task design. Task allocation plays a dual role: it determines both the value of the information produced by the team and the effectiveness of peer monitoring. To model the value of information, we suppose the principal must make a decision at every location in the task space and incurs a loss when her decision differs from the local state. The first-best allocation spreads agents across the task space to reduce overlap in their expertise and improve informational coverage.

Robust incentives push in the opposite direction. Incentive costs are lower when agents are closer to their monitors and closer to the principal. Deterministic task assignment therefore exposes a trade-off between diversity in knowledge production and agency costs. Because the hierarchy characterized above makes robust incentive cost separable across neighboring monitoring links, this trade-off can be expressed tractably. The optimal deterministic allocation still equalizes neighboring distances, but the common distance is smaller than in the first best.  
Robust incentives therefore distort knowledge production toward the principal’s core expertise: the organization over-invests in familiar, easily monitored parts of the project and under-explores more novel regions of the task space. This concentration result is a consequence of robust implementation rather than a generic implication of moral hazard. Under the Bayesian Nash benchmark, the contract can use cyclic peer monitoring among agents and task allocation can be distorted in qualitatively different ways.

A transparent deterministic assignment leads to distortion because it cannot separate the two roles of task allocation: knowledge production and incentive provision. On the other hand, confidential random assignment can decouple these roles. Because each agent observes only her own realized task, she might remain uncertain about whether other agents have been assigned to nearby tasks, and hence about the realized monitoring hierarchy.

With high probability, agents are assigned to first-best locations, which preserves broad exploration. With small probability, some agents are assigned to colocated or closely related tasks and are promised large bonuses. Even an agent assigned to her first-best location cannot rule out being in a rare monitoring state in which a nearby peer can credibly evaluate her. As these rare monitoring states become less likely and the associated bonuses grow, expected payments approach the minimum possible incentive cost while information loss approaches the first best. Therefore, confidentiality relaxes the hierarchy-induced distortion by allowing the principal to design not only who monitors whom, but also what monitoring hierarchy agents perceive as possible.

\paragraph{Related literature.} 
Our paper contributes to the literature on peer monitoring in information production and elicitation. A common theme is that when agents acquire correlated information, compensation based on agreement in reports can be used to provide incentive.  
See, for example, \cite{maskin1999nash}, \cite{pesendorfer2003second}, \cite{bohren2019peer}, and \cite{azrieli2021monitoring}.  
However, agreement-based mechanisms can induce strategic risk, which makes information production vulnerable to self-fulfilling multiplicity. We therefore study how contracts can eliminate this risk by making the desirable outcome uniquely rationalizable, and characterize the resulting monitoring structure and task allocation. A related concern is emphasized by \cite{pei2025robust}, who show that agreement-based elicitation schemes can be fragile under a different notion of robustness based on small perturbations of agents’ preferences.

We join the growing literature on robust incentive provision in team production, where agents’ incentives are pinned down sequentially through an iterated-elimination logic. In \cite{winter2004incentives} and \cite{halac2021rank}, agents’ efforts jointly determine team success, which is the only contractible performance measure. These papers study how such technological complementarities shape the robust compensation structure, including differential pay and wage confidentiality. Recent work on robust contracting in team production assumes that the principal has access to a rich set of contractible monitoring signals. See, e.g., \cite{cusumano2023misaligning}, \cite{halac2023monitoring}, and \cite{camboni2023monitoring}. In contrast, we abstract from production externalities and instead focus on informational complementarities arising from task similarity. Information acquired while performing one task is also informative about related tasks, allowing the principal to design peer-monitoring relationships. This extends the robust contracting framework to encompass the endogenous design of monitoring authority and yields new insights into the optimal organization of peer monitoring.

Also related is the spatial learning literature pioneered by \cite{jovanovic1990long}, \cite{callander2011searching} and \cite{garfagnini2016social}. 
Recent developments (e.g., \cite{bardhi2022attributes}, \cite{Arjadanina}, \cite{dong2023does}, and \cite{aybas2023efficient}) study various agency problems under spatial learning. See \cite{bardhi2026learning} for a survey. 
Because our goal is to model how task relatedness shapes peer monitoring, rather than the full complexity of the learning process, we depart from the Gaussian-process approach common in this literature and introduce a simpler Markov-chain specification.

Finally, our paper complements the literature on knowledge hierarchies in organizations, initiated by \cite{garicano2000hierarchies}. That literature studies how organizations allocate problem-solving authority when knowledge is dispersed: lower-level agents solve routine problems, while rarer and more difficult cases are referred upward. Our hierarchy is different. Tasks differ not in problem frequency, but in how effectively they can be monitored. The central question is how to provide incentives when knowledge generated on one task helps evaluate nearby tasks. Thus, we study a hierarchy of monitoring authority, in which higher-level agents influence the compensation of lower-level agents through their reports. This highlights a distinct role of task allocation: it shapes not only who solves which problems, but also who can credibly monitor whom.

\section{Model}\label{sec_model}

A principal organizes a team of agents to work on related knowledge-production tasks. He both manages the team and participates directly in production. By working on her own task, an agent produces task-specific information that is also useful for evaluating the quality of others' work.
The principal's problem is to use these informational spillovers and design monitoring and incentives.

\subsection{Environment}

\paragraph{Agents and tasks.} We represent tasks as locations on the interval $[0,1]$. Each location $t\in[0,1]$ is associated with a binary local state $x(t)\in\{A,B\}$. Agent $i\in\mathcal{N}\triangleq\{1,\ldots,n\}$ is assigned to location $t_i\in[0,1]$. Agent $i$'s knowledge-production task is to learn the local state $x(t_i)$, which requires costly effort. Without loss of generality, we assume $t_1\leq\cdots\leq t_n$ and denote $\tau=(t_1,...,t_n)$ as agents' task profile. 
The main analysis takes task allocation as exogenous; Section~\ref{sec_endo} endogenizes it.

To capture task relatedness, we model $x(\cdot)$ as the realized path of a Markov process on $[0,1]$, so that nearby locations are more likely to have the same local state. Specifically, define the overall \emph{state} $x:[0,1]\rightarrow \{A,B\}$ 
as a realized path of the following Markov process: 
The initial state, $x(0)$, is a uniform random draw from $\{A,B\}$. Moving along the interval, switches in the state occur according to a Poisson process with rate $\lambda>0$: at each location $t$, the distance to the next shock is exponentially distributed with mean $1/\lambda$. 
If no shock occurs at $t$, the path $x(\cdot)$ is continuous at that location. If a shock arrives, then $x(t)$ is drawn independently and uniformly from $\{A,B\}$. See Figure \ref{fig_state} for the illustration of a realized state. The distribution of the state $(x(t))_{t \in [0,1]}$ is commonly known, and through its marginal distributions, it induces the distribution of local states relevant to agents, $(x(t_i))_{i=1}^n$.

\begin{figure}[ht]
\begin{center}
\footnotesize
\begin{tikzpicture}[scale=4.5]

\node at (1.95,0) {$t$};

\node at (1.8,-0.05) {$1$};
\node at (0,-0.05) {$0$};
\node at (0.5, -0.05) {$t'$};
\node at (1, -0.05) {$t''$};
\node at (1.4, -0.05) {$t'''$};

\node at (0,0.85) {$x(t)$};
\node at (-0.05,0.6) {$B$};
\node at (-0.05,0) {$A$};

\draw [->, thick,gray] (0,0) -- (0,0.8);

\draw [->, thick,gray] (0,0) -- (1.9,0);

\draw [ultra thick,nberblue] (0,0.6) -- (0.5,0.6);
\draw [gray, dashed] (0.5,0) -- (0.5,0.6);
\draw [ultra thick,nberblue] (0.5,0) -- (1,0);
\draw [gray, dashed] (1,0) -- (1,0.6);
\draw [ultra thick,nberblue] (1,0.6) -- (1.8,0.6);
\draw [gray, dashed] (1.4,0) -- (1.4,0.6);

\end{tikzpicture}

\end{center}
\caption{Illustration of the State}
\bigskip
{\scriptsize Notes: In this example, $A$ and $B$ are real numbers such that $A<B$. The state realization is such that $x(t) = A$ if $t \in [t', t'')$ and $x(t) = B$ otherwise. It corresponds to any realized path of the Markov process such that the initial state is $x(0)=B$ and state transitions occur at $t'$ and $t''$. In this example, shocks arrive at $t', t''$, and $t'''$, but state transitions occur only at $t'$ and $t''$.}

\label{fig_state}
\end{figure}
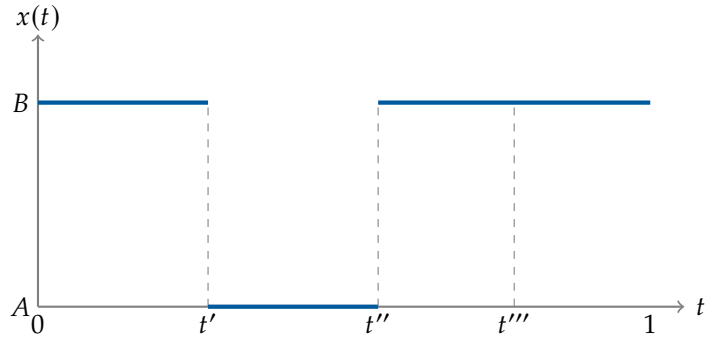

Modeling the state as a realized path of a Markov process allows us to formalize the \emph{similarity of tasks}. 
As we show in Section \ref{sectionKey}, closer locations are more likely to have the same local state.
Thus, information generated by agent $i$’s work is more informative about the local states of nearby tasks.

We assume that the principal observes the realized state $x(0)$ for $t=0$ at no cost and can use it for contracting, but he does not directly observe any other local states.
Location $0$ thus serves as a benchmark problem the principal can handle directly, while locations farther from $0$ represent tasks that are less familiar to the principal. This assumption captures environments in which the principal not only manages the organization but also has direct expertise in a core part of its work, such as a lead engineer or principal investigator.

Although the underlying economics we want to highlight does not require placing the principal’s expertise at $t=0$, this normalization yields a natural ordering of tasks by their proximity to the principal’s expertise: 
the principal’s information is more informative about tasks closer to $0$ and less informative about tasks farther away. The ordering $t_1\leq\cdots\leq t_n$ therefore ranks agents by how easily the principal can monitor their tasks. 
This direct monitorability is the only source of heterogeneity across agents.\footnote{With an interior location for the principal, agents on each side could be ordered by distance from the principal, and the subsequent analysis would apply separately to the two sides.}

\paragraph{Information production.}

Agents can exert costly effort to discover realized local states of their assigned tasks. Specifically, if agent $i$ exerts effort, she incurs a cost of $c > 0$ and produces a signal $\signal_i$ that equals the local state of her location, i.e., $\signal_i = x(t_i)$.\footnote{ We assume that agents can exert effort only for tasks to which they are assigned. This naturally fits applications such as equipment maintenance, software debugging, auditing, and experimental research, where information acquisition is tied to task-specific access, equipment, resources, or direct participation.}  
If agent $i$ shirks, which is costless, her signal \(\signal_i\) is independent of the state and instead determined by a third party, called the \emph{noisemaker}.

The noisemaker is a player who chooses a noise profile \((\shirk_1, \dots, \shirk_n) \in \{A, B\}^n\), where each \(\shirk_i\) represents the realization of the signal produced by agent $i$ if (and only if) shirking.  
The noisemaker is indifferent between all possible noise profiles in $\{A, B\}^n$ (i.e., its payoff equals $0$ for any strategy profile) and moves simultaneously with agents.

The noisemaker is a modeling device that captures a fundamental ambiguity: when an agent shirks, neither the principal nor other agents can predict—or even form beliefs about—the signal she generates (or fabricates). One can interpret the noisemaker’s output as a shirker’s random, uninformed guess about her local state.
Because our objective is to implement the outcome in which all agents exert effort, we model a shirker's signal as arbitrary off path. At the same time, assigning this off-path signal to the noisemaker keeps each agent's strategy space binary, which facilitates the characterization of robust contracts.
Section \ref{sectionKey} provides a detailed discussion.

\paragraph{Contract.} 
The realized state $x$, agents' actions, and the noisemaker's choice $(\shirk_1,..., \shirk_n)$ determine a \emph{signal profile}, denoted by $
\signals=(\signal_0, \signal_1,..., \signal_n) \in \{A, B\}^{n+1}$. Here, $\signal_0$ denotes the principal's signal, which equals the local state $x(0)$ at $t=0$.
Other signals are agents', where $\signal_i = \state(t_i)$ if agent $i$ exerts effort and $\signal_i =y_i$ if she shirks.
A signal profile $\signals$ is \emph{truthful} if $\signals$ is a result of all agents exerting effort, i.e., $\signals = (x(0), x(t_1),..., x(t_n))$.
A truthful signal profile reveals the local states of at most $n+1$ locations.

The principal commits to how to compensate agents based on a realized signal profile.
Formally, the principal chooses a \emph{contract}, denoted by 
$
w = (w_1(\signals),..., w_n(\signals))_{\signals \in \{A, B\}^{n+1}}$, 
Under this contract, agent $i$ receives $w_i(\signals) \ge 0$ when the signal profile is $\signals$. 
We assume limited liability, so payments to agents must be nonnegative.\footnote{We do not explicitly model the agent's participation decision, but assuming a zero outside option ensures that each agent weakly prefers accepting the assigned task and choosing $a_i = 0$ to opting out.}

The signal profile is the only contractible object. This reflects the idea that in knowledge-production environments, the quality of task performance is often difficult for outsiders to verify directly. Relevant contingencies are local and task-specific, so whether a task has been performed well is often best assessed by someone working on a related problem. Incentive provision must therefore rely on the informational relationship across tasks. The central challenge is that monitoring ability is itself generated by work on related tasks.

\paragraph{Payoffs.} 
Agents are risk-neutral:
The payoff of an agent is equal to the payment from the principal minus effort cost $c$ if she exerts effort; otherwise, the payoff is equal to the payment.

\paragraph{Robust contract design.} 
Each contract $w$ induces an incomplete-information game, denoted by $\Gamma(w,\tau)$. In this game, each agent $i$ chooses $a_i\in\{0,1\}$, where $a_i=1$ denotes effort and $a_i=0$ denotes shirking, while the noisemaker chooses a profile $(\shirk_1,..., \shirk_n) \in \{A, B\}^n$, where $\shirk_i$ is the signal assigned to agent $i$ if she shirks. Payoffs are as described above. 

The principal's goal is to choose a least-cost contract that induces all agents to work in every correlated rationalizable strategy profile of $\Gamma(w, \tau)$ \citep{fudenberg1991game}.
We call any such contract \emph{robustly implements effort (RIE)} and write $\contracts(\tau)$ for the set of all RIE contracts under task profile $\tau$.
An RIE contract eliminates strategic risk in effort provision and only relies on agents' rationality and beliefs about their opponents' rationality, and so on.

The principal's problem is to choose a contract $w\in \contracts(\tau)$ to minimize total expected payment. Formally, given a contract $w$, the total expected payment when all agents exert effort is 
\[\hat K(w, \tau) \triangleq \sum_{i=1}^n \sum_{\signals \in \{A, B\}^{n+1}} w_i(\signals)\Pr(\signals|\tau),\]
where $\Pr(\signals|\tau)$ is the probability of truthful signal profile $\signals$ under task profile $\tau$.
The optimal robust contracting problem is then 
\begin{align}\label{inv}\tag{P-C}
K(\tau)\triangleq\inf_{w\in\contracts(\tau)}\hat K(w,\tau).
\end{align}
The problem \eqref{inv} generally does not have a minimum. We call the infimum of total expected payment, $K(\tau)$, the \emph{incentive cost} given task profile $\tau$, and we study approximately optimal contracts. That is, we look for a sequence $\{w^m\}$ of contracts such that each $w^m$ robustly implements effort and $\lim_{m\rightarrow\infty}\hat K(w^m,\tau)= K(\tau)$. With a slight abuse of notation, we call $w^*=\lim_{m\rightarrow\infty}w^m$ an \emph{optimal contract}. Note that an optimal contract itself may not be a RIE.

By the standard argument \citep{bernheim1984rationalizable,pearce1984rationalizable}, a contract $w$ belongs to $\contracts(\tau)$ if and only if all agents exert effort in every outcome of $\Gamma(w,\tau)$ that survives iterated elimination of strictly dominated strategies (IESDS). The noisemaker is indifferent among all choices in $\{A,B\}^n$, so its strategies are never eliminated in this process. We therefore use IESDS as an equivalent notion of robust implementation.

\subsection{Discussion of Assumptions}
\label{sectionKey}
Before moving forward, we discuss some key assumptions of the model. 

\paragraph{State distribution.}
Our Markov process formulation of the state has several useful properties.
First, the local state of one location is informative about other local states.
To see this, take two locations, $t$ and $t'>t$.
We have $x(t)\neq x(t')$ if and only if, in the interval $(t,t']$, (i) at least one shock arrives, which occurs with probability $\int^{t'-t}_0\lambda e^{-\lambda s}ds$, and (ii) the realized local state at the last 
shock differs from $x(t)$, which occurs with probability $1/2$.
As a result, the two local states take the same value with probability 
\begin{equation}\label{cor}
\Pr[x(t)= x(t')]=1- \frac{1}{2}\int^{t'-t}_0\lambda e^{-\lambda s}ds
= \frac{1}{2} + \frac{1}{2}e^{-\lambda(t'-t)}, 
\end{equation}
which is strictly greater than $1/2$. Thus, local states are positively correlated across locations. Moreover, this correlation decreases with 
$|t-t'|$: the state at a nearby location provides a more informative signal about an agent’s task than the state at a distant one. 
Thus, the local state $x(t)$ provides a more precise monitoring signal for the performance of an agent assigned to task $t'$ when the two tasks are closer. The decline is steeper at short distances and flatter at long distances, so the informational effect of proximity is strongest locally.

Second, the local states of the nearest sampled locations are sufficient statistics for an unsampled location:
Formally, suppose that the principal observes the local states of locations $t_1<\cdots < t_m$.
The principal's belief over the local state at any unsampled location $t\in (t_j, t_{j+1})$ is given by
$\Pr\big(x(t)=A|\{x(t_i)\}_{i=1,...,m}\big)=\Pr\big(x(t)=A|x(t_j), x(t_{j+1})\big)$, 
because the state $x(\cdot)$ is a realized path of a Markov process. 

Third, we assume that if a shock arrives at $t$, then $x(t)$ is a uniform draw from $\{A,B\}$. Thus, conditional on a shock at $t$, the distribution of local state $x(t)$ is orthogonal to the local states of location $s<t$. This independence would fail if the local state deterministically switched upon each shock, i.e., conditional on a shock at $t$, $x(t)\neq x(t^-)$ with probability $1$. 

Finally, we adopt a binary-state Markov process to capture environments in which local correlation arises from regime persistence. The underlying state is locally stable and changes only when the task environment crosses a discrete boundary. Thus, nearby tasks share the same regime unless such a boundary lies between them. This specification captures a form of uncertainty in which the relevant task environment is piecewise stable rather than smoothly varying. It differs from the Brownian-motion formulation prevalent in spatial learning (see, e.g., \cite{callander2011searching} and \cite{bardhi2026learning}), which is natural when local correlation comes from diffusive accumulation: underlying conditions evolve continuously, nearby tasks are similar because incremental variance over short distances is small, and large changes over short distances are unlikely.

Accordingly, our model is naturally applied to environments characterized by persistent local regimes and occasional discrete breaks. For example, in diagnosing a production, maintenance, or software system, most adjacent components may operate normally until one reaches the point at which a failure or defect begins; after that boundary, nearby components are likely to share the faulty regime. Similar patterns arise in organizational audits, where compliance problems may be concentrated within a unit or process; in experimental research, where a target property may appear only within a particular range of temperatures or compositions; or in project teams, where a common coordination shock may impair a cluster of tasks. In such settings, nearby tasks are informative
because regimes tend to persist locally and change only at discrete boundaries.

\paragraph{Noisemaker.} We use the noisemaker to capture realistic deviations from the desired behavior while keeping each agent's action binary. Effort represents acquiring information and reporting it truthfully. Shirking represents departures from this benchmark, whether through a failure to acquire information, strategic misrepresentation, or both. The noisemaker summarizes the informational consequences of these deviations for other agents and the principal. Because these consequences need not be known or commonly understood off path, we impose no known distribution on the noisemaker’s signals. This preserves the relevant strategic uncertainty while keeping the IESDS characterization tractable.

\section{Optimal Contract}

\subsection{Monitoring Hierarchies}

We decompose the contract design problem into two steps. The first is \emph{monitoring structure design}: for each agent, determine whose signals are used to evaluate her performance. The second is \emph{compensation design}: given a monitoring structure, choose payments to implement effort at lowest cost.

This two-step decomposition serves two purposes.
Economically, it yields a natural language for describing the peer-monitoring relationships induced by a contract and for discussing how RIE requirement shapes the structure of peer monitoring. Analytically, it reformulates the original high-dimensional contract design problem into a sequential one, in which each step involves a lower-dimensional choice. 

We now introduce notation to formalize a monitoring structure and then show that, for the principal's cost-minimization problem, it is without loss to focus on hierarchical monitoring structures.

\begin{definition}
Fix any contract $w$.
Node $j\in \mathcal N\cup\{0\}$ \emph{monitors} (or has \emph{monitoring authority} over) agent $i\in\mathcal N$, with $j\neq i$, if $j$'s signal affects $i$'s payment; that is, if, for $s_j=A$ and $s'_j = B$, there exists a profile of the remaining signals $s_{-j}$ such that
$
w_i(s_j,s_{-j})\neq w_i(s_j',s_{-j}).
$
The \emph{monitoring structure} induced by contract $w$ is the directed graph $G(w)$ on $\mathcal N\cup\{0\}$ with an edge $j\to i$ whenever $j$ monitors $i$.
\end{definition}

Monitoring authority captures an agent’s ability to influence another agent’s compensation through her signal.
In general, monitoring authority between any pair of agents may be mutual (i.e., both $i\to j$ and $j\to i$ are in $G(w)$), one-sided, or absent.
The monitoring structure records, for each agent $i$, the set of agents over whom agent $i$ has such authority. 
We are particularly interested in monitoring structures induced by RIE contracts.

Monitoring authority allows us to connect peer-monitoring structure with robust implementation.
The IESDS characterization of robust implementation implies that any RIE contract induces an ordering in which agents' shirking strategies are eliminated.
Agents whose shirking strategies are eliminated earlier must work independently of how later-round agents behave.
Their signals can therefore serve as credible sources of information for disciplining agents whose incentives are secured only in subsequent rounds.
This ordering does not by itself require the original graph $G(w)$ to be hierarchical, because $G(w)$ records all payment dependence, including dependence that is unnecessary for incentives.
The result below shows that, for cost minimization, such unnecessary dependence can be removed: Any RIE contract can be replaced by a weakly cheaper RIE contract in which monitoring authority flows only from earlier to later rounds of the elimination order.

To formalize this idea, we first define a class of contracts under which we can order agents according to their monitoring authority.

\begin{definition}\label{definition1}
	A monitoring structure $G$ is \emph{hierarchical} if there exists an ordered partition $H=(h_1,\dots,h_m)$ of the set $\mathcal N$ of agents that satisfies the following: (i) for every edge $j\to i$ in $G$, either $j=0$ or $j\in h_{\ell'}$ and $i\in h_\ell$ for some $\ell'<\ell$; and (ii) every agent is reachable from the principal through a directed path in $G$. 
\end{definition}
Under a hierarchical monitoring structure, monitoring can flow only from the principal or from a higher layer to a lower layer.
The hierarchy restricts the direction of permissible monitoring links but does not restrict which permissible links are used.\footnote{For example, if $h_1 = \{1\}$, $h_2 = \{2\}$, and $h_3 = \{3\}$, agent $3$'s compensation may depend on agent $1$'s signal even when it does not depend on agent $2$'s signal, i.e., $1 \to 3$ may be in $G$ even if $2 \to 3$ is not. Conversely, Definition \ref{definition1} does not require either link to be in $G$.}
Monitoring hierarchies can take a variety of forms. At one extreme is a \emph{star hierarchy}, in which all agents are placed in a single layer and monitored directly by the principal, corresponding to the classical monitoring story of \citet{alchian1972production}. 
At the other extreme is a \emph{chain hierarchy}, in which agents are ordered sequentially and each layer contains a single agent.

\begin{lemma}\label{prop-RIE}
	For any hierarchical monitoring structure, there is an RIE contract that induces the same monitoring structure.
For any RIE contract, there exists an RIE contract that induces a hierarchical monitoring structure and attains weakly lower expected payments.
\end{lemma}

To show the first part, take any hierarchical monitoring structure $G$ with associated ordered partition $(h_1,\ldots,h_m)$.
Then, we can construct an RIE contract recursively, level by level. For each agent $j\in h_1$, let her compensation depend only on $(s_0,s_j)$, so that her incentive is secured directly by the principal's signal. 
Because $x(0)$ and $x(t_j)$ are correlated, a sufficiently large reward for agreement between $s_0$ and $s_j$ guarantees incentive compatibility. Then proceed inductively: for each agent $j\in h_i$, let $M_j=\{k\in \mathcal N\cup\{0\}: k\to j \text{ in } G\}$ be the set of $j$'s monitors in the given hierarchical monitoring structure, and design $w_j$ to depend only on her own signal and the signals of agents in $M_j$, with a strictly positive agreement bonus attached to each prescribed monitor's signal.
Because the incentives of agents in $h_1,\ldots,h_{i-1}$ have already been secured, the prescribed monitors' signals can be used in turn to incentivize agents in $h_i$. Repeating this argument level by level yields an RIE contract that induces the hierarchy.

The second part follows from the fact that, given any RIE contract, the principal can lower the expected payment to each agent $i$ by making her compensation independent of the signals of agents whose incentives are secured in later rounds of the IESDS procedure. 
RIE requires agent $i$ to work even when all later-round agents shirk and the noisemaker selects arbitrary signals for them. Holding fixed agent $i$'s own signal and all earlier-round signals, one can therefore replace $w_i$ at every realization of later-round signals by the lowest value across those realizations: this preserves $i$'s incentive to work and weakly reduces the expected payment. 
Hence, in a cheapest RIE contract, $w_i$ cannot depend on later-round signals. 
Applying this argument for each $i$ yields a hierarchical monitoring structure in which each element $h_k$ of the ordered partition consists of the agents who delete shirking at the $k$-th round of IESDS.

Lemma \ref{prop-RIE} provides a way to solve the principal's high-dimensional design problem. The key step is to determine the monitoring hierarchy: which agents should occupy each layer of monitoring authority, and in what order should their incentives be secured? Once a hierarchy is fixed, the remaining problem is to determine which higher-layer signals should be used to discipline each agent and then to find the least-cost contract consistent with it. Specifically, for each agent $j\in h_i$, $w_j(\cdot)$ may depend only on the principal's signal, agent $j$'s own signal, and the signals of agents in higher layers $h_{i'}$ with $i'<i$; moreover, agent $j$ must prefer to work given the belief that all agents in higher layers exert effort.

\subsection{A Two-Agent Example}
We begin with the two-agent case under $t_1<t_2$  to illustrate the main economic forces behind the optimal monitoring hierarchy. 
In this case, there are three possible monitoring hierarchies: the single-layer hierarchy $(\{1,2\})$, and two-layer hierarchies, $H_1=(\{1\},\{2\})$ and $H_2=(\{2\},\{1\})$. 

We proceed in three steps. 
We first show that the single-layer hierarchy is suboptimal. We then derive the optimal contract consistent with $H_1$ and that consistent with $H_2$.
Finally, we compare the resulting incentive costs and show the optimum robust contract is consistent with $H_1$, i.e., the principal monitors agent 1, and agent 1 monitors agent 2.

\paragraph{Single-layer is suboptimal.} 
Under $(\{1,2\})$, the principal monitors both agents directly.
This is suboptimal.
The signal of agent $2$ is more strongly correlated with the signal of agent $1$ than that of the principal, because 
agent $2$ is closer to agent $1$ than to the principal.
Thus, basing agent $2$'s pay on agent $1$'s signal induces agent $2$ to work at a lower cost than when agent $2$'s pay depends on a comparison between his signal and the principal's signal.
It is therefore better to use the principal's signal to secure agent 1's incentives, and then use agent 1's signal to discipline agent 2.

Formally, under the single-layer structure, each agent $i$'s payment depends only on $(s_0,s_i)$. 
Because agents are risk neutral, and the event $s_i=s_0$ is the most informative signal of work relative to shirk, it is optimal to pay agent $i$ a positive amount $w_i$ only when her signal agrees with the principal's signal. Agent $i$'s (strict) incentive constraint is
\[
\frac12(1+e^{-\lambda t_i})w_i-c > \frac12 w_i,
\]
where the right-hand side is the agent's shirking payoff, because she can shirk and still match the principal's signal with probability $1/2$.
The infimum payment that satisfies the constraint is $w_i=2ce^{\lambda t_i}$. The corresponding incentive cost under the single-layer hierarchy is therefore
\[
c(1+e^{\lambda t_1})+c(1+e^{\lambda t_2}).
\]
Under $H_1=(\{1\},\{2\})$, the cost of securing agent 1's incentive remains $c(1+e^{\lambda t_1})$, while the cost of securing agent 2's incentive falls to $c(1+e^{\lambda (t_2-t_1)})$. Hence, any contract that induces the single-layer hierarchy is strictly dominated.

\paragraph{Chain dominates sandwich.} We can further narrow down the candidate optimal designs within each of the two-layer hierarchies.  Under $H_1$, agent 1 is monitored directly by the principal, and agent 2 is then disciplined using the signals available from the higher layer. Because agent 1's signal is a sufficient statistic for $(s_0,s_1)$ in monitoring agent 2, it is without loss to focus on contracts that induce a \emph{chain} monitoring structure, in which the principal monitors agent 1 and agent 1 monitors agent 2. 
As shown above, the cost infimum for these contracts is
\begin{equation}\label{cost-chain}
K^c(t_1,t_2)\triangleq c\bigl(1+e^{\lambda t_1}\bigr)+c\bigl(1+e^{\lambda (t_2-t_1)}\bigr).
\end{equation}
By contrast, under $H_2$, agent 2 is monitored directly by the principal, and agent 1 is then disciplined using the principal's signal together with agent 2's signal. In this case, neither signal is sufficient for the other in monitoring agent 1, so it is without loss to focus on contracts that induce a \emph{sandwich} monitoring structure, in which agent 2 is monitored by the principal while agent 1 is monitored using both signals. The least incentive cost for such contracts is
\[
K^s(t_1,t_2)\triangleq
c\left[
1+\frac{1+e^{-\lambda t_2}}{e^{-\lambda t_1}+e^{-\lambda (t_2-t_1)}}
\right]
+
c\bigl(1+e^{\lambda t_2}\bigr).
\]
Figure~\ref{fig:chain_sandwich} illustrates the chain and sandwich structures induced by $H_1$ and $H_2$, respectively.

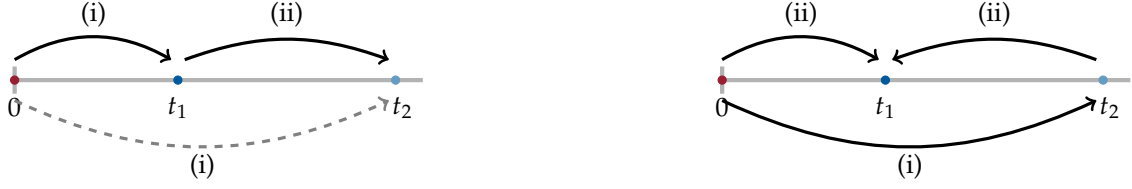
\begin{figure}

\begin{center}
\begin{subfigure}[b]{0.45\textwidth}
\centering
\begin{tikzpicture}[scale=1.8]
\footnotesize

% labels
\node at (0, 1.3) {$0$};
\node at (1.2, 1.3) {$t_1$};
\node at (2.85, 1.3) {$t_2$};

% line
\draw [ultra thick, gray!60] (0,1.5) -- (3,1.5);
\draw [ultra thick, gray!60] (0,1.4) -- (0,1.6);
\draw [fill, pennred] (0,1.5) circle [radius=0.03];
\draw [fill, nberblue] (1.2,1.5) circle [radius=0.03];
\draw [fill, nberblue!60] (2.8,1.5) circle [radius=0.03];

% arrows
\draw[very thick,->] (0,1.65) to [out=30,in=150] (1.15,1.65);
\node[above=8pt] at (0.58,1.67) {(i)};

\draw[very thick,->] (1.25,1.65) to [out=20,in=160] (2.75,1.65);
\node[above=8pt] at (2.0,1.67) {(ii)};

% redundant direct monitoring from principal to agent 2
\draw[very thick,->,dashed,gray] (0,1.35) to [out=-25,in=-155] (2.75,1.35);
\node[below=14pt] at (1.38,1.30) {(i)};
\end{tikzpicture}
%\caption{Chain}
\end{subfigure}
\hfill
\begin{subfigure}[b]{0.45\textwidth}
\centering
\begin{tikzpicture}[scale=1.8]
\footnotesize

% labels
\node at (0, 1.3) {$0$};
\node at (1.2, 1.3) {$t_1$};
\node at (2.85, 1.3) {$t_2$};

% line
\draw [ultra thick, gray!60] (0,1.5) -- (3,1.5);
\draw [ultra thick, gray!60] (0,1.4) -- (0,1.6);
\draw [fill, pennred] (0,1.5) circle [radius=0.03];
\draw [fill, nberblue] (1.2,1.5) circle [radius=0.03];
\draw [fill, nberblue!60] (2.8,1.5) circle [radius=0.03];

% arrows
\draw[very thick,->] (0,1.35) to [out=-25,in=-155] (2.75,1.35);
\node[below=14pt] at (1.38,1.30) {(i)};

\draw[very thick,->] (0,1.65) to [out=30,in=150] (1.15,1.65);
\node[above=8pt] at (0.58,1.67) {(ii)};

\draw[very thick,<-] (1.25,1.65) to [out=20,in=160] (2.75,1.65);
\node[above=8pt] at (2.0,1.67) {(ii)};

\end{tikzpicture}
%\caption{Sandwich}
\end{subfigure}
\end{center}

\caption{Monitoring structures under the two two-layer hierarchies}
\bigskip
{\scriptsize Notes: In the left panel, agent 1's signal is sufficient for monitoring agent 2, so using the principal's signal to monitor agent 2 (dashed curve) is redundant. In the right panel, both the principal's signal and agent 2's signal are used to monitor agent 1. Roman numerals denote the order of elimination in the IESDS process.}
\label{fig:chain_sandwich}
\end{figure}

What remains is to compare $K^c(t_1,t_2)$ and $K^s(t_1,t_2)$. The comparison reflects the following trade-off. 
Under a chain, the principal first disciplines the agent who is easiest to monitor directly, which minimizes the cost of establishing the first credible peer signal in the hierarchy. Under a sandwich, by contrast, the principal first disciplines the more distant agent 2, which is more expensive, in order to obtain an additional signal for monitoring agent 1 in the second step. 
Thus, the sandwich trades a higher first-layer cost for a richer second-layer monitoring environment. 

We show that the chain structure attains a lower total payment.
To see this, fix $t_2$ and view $K^c(t_1,t_2)$ and $K^s(t_1,t_2)$ as functions of $t_1$. Under the chain structure, each agent is monitored only by her immediate predecessor: agent 1 by the principal, and agent 2 by agent 1. Hence, the incentive cost is the sum of two single-link monitoring costs as in equation \eqref{cost-chain}. 

The key observation is that each single-link cost is increasing and convex in the monitoring distance. 
The reason is that the informativeness of a monitor's signal decays exponentially with distance.\footnote{A similar property appears in other specification. See, e.g., \cite{bardhi2022attributes} and \cite{Arjadanina}.} 
When the monitoring distance is \(d\), effort raises the probability that the agent's signal agrees with her monitor's signal by an amount proportional to \(e^{-\lambda d}\). 
To compensate the agent for effort cost \(c\), the required bonus is therefore proportional to $\frac{1}{e^{-\lambda d}}=e^{\lambda d}$. 
Thus, as the monitoring distance increases, the required payment rises at an increasing rate. 

Holding \(t_2\) fixed, the chain cost is now the sum of two convex terms: one in \(t_1\) and one in \(t_2-t_1\). 
Moving \(t_1\) toward the midpoint \(t_2/2\) shortens the longer monitoring link and lengthens the shorter one by the same amount. 
By convexity, this lowers total cost until the two monitoring distances are equal. 
Hence, \(K^c(\cdot,t_2)\) is convex and minimized at \(t_1=t_2/2\).

Under the sandwich structure, the expected payment to agent 1 depends on the joint use of the two signals and therefore takes a multiplicative form in $\Pr(s_1=s_0)$ and $\Pr(s_1=s_2)$. Fix $t_2$. When $t_1=0$ or $t_1=t_2$, agent 1 coincides with either the principal's task or agent 2's task, so one of the two monitoring links is maximally informative. This makes agent 1 easiest to monitor and minimizes the incentive cost. As $t_1$ moves from either endpoint toward the midpoint $t_2/2$, the initially very precise monitoring link becomes less informative, while the other link improves by less because informativeness falls at a decreasing rate with distance. 
Thus, combined monitoring becomes weaker and the required payment rises.
By symmetry, this effect is strongest at $t_1=t_2/2$, where monitoring is least effective. Therefore, the monitoring benefit is maximized at the endpoints and minimized at the midpoint, implying that $K^s(\cdot,t_2)$ is single-peaked.

However, at the boundary $t_1=0$, the two structures coincide, i.e., $K^s(0,t_2)=K^c(0,t_2)$. This is because when agent 1 is colocated with the principal, agent 2's signal provides no additional value in disciplining agent 1 and thus the sandwich collapses to the chain. Because $K^s(\cdot,t_2)$ attains its minimum at the boundary while $K^c(\cdot,t_2)$ attains its maximum there, it follows that
\[
K^s(t_1,t_2)\ge K^s(0,t_2)=K^c(0,t_2)\ge K^c(t_1,t_2),
\]
for all $t_1\le t_2$. Thus, the chain weakly dominates the sandwich.

\paragraph{Wedge between informational and monitoring authority.} Robust implementation requires monitoring authority to flow downward, from higher layers to lower layers.  As a result, it does not fully use the information available for incentive provision. For example, agent \(1\)'s pay does not depend on agent \(2\)'s signal, even though agent \(2\)'s signal is informative about agent \(1\)'s performance along the target outcome. This separates informational authority from  monitoring authority. When all agents work, lower-layer agents may acquire information useful for evaluating agents above them. From the perspective of the sufficient-statistic principle \citep{holmstrom1982moral}, such information should enter compensation. The obstacle is circular incentive dependence. 
Upward or mutual monitoring would make agents' incentives rely on the credibility of reports from agents whose own incentives have not yet been secured. This can generate additional equilibria, including coordinated shirking, and runs counter to RIE’s conservative approach to strategic risk (see Section \ref{sec:BNI} for details).

\subsection{Optimal RIE in General Case}

We now turn to the case of an arbitrary number $n$ of agents. Although we can simplify the principal's problem by restricting attention to contracts that induce hierarchical monitoring hierarchies, the number of candidate hierarchies is still large.
Specifically, the total number of hierarchies equals the total number of order partitions of the set $\{1,..., n\}$, which is 
$
\sum_{k=1}^n k!\,S(n,k),
$ 
where $S(n,k)$ is the Stirling number of the second kind. For example, when $n=8$, there are $273{,}343$ possible hierarchies. However, the insight from the two-agent case allows us to focus attention on a smaller class of monitoring hierarchies.

We first consider a class of contracts that can be used to construct a sequence of RIE contracts that approximates the incentive cost.

\begin{lemma}\label{lem_chain_contract}
For any $\epsilon>0$, define the contract $w^\epsilon$ by
\begin{equation}\label{wage}
w^\epsilon_i(\signals) \triangleq 
\begin{cases}
2c e^{\lambda (t_i-t_{i-1})}+\epsilon & \mbox{if } s_i=s_{i-1},\\
0 & \mbox{if } s_i\neq s_{i-1},
\end{cases}
\end{equation}
for each agent $i=1,\dots,n$, where $t_0=0$. Then $w^\epsilon$ is an RIE contract and induces a monitoring hierarchy such that $h_i=\{i\}, \forall i$.
\end{lemma}

The logic is immediate because agents' incentives are secured sequentially along the chain: 
Agent $1$ is incentivized directly by comparison with the principal's signal, and each subsequent agent $i$ is incentivized by comparison with the signal of agent $i-1$, whose incentive has already been secured. 
The strictly positive slack $\epsilon$ guarantees strict incentive compatibility at each step.

The contract $w^\epsilon$ induces a monitoring hierarchy in which monitoring authority declines with distance from the principal. 
Under this contract, each agent $i$ is rewarded when her signal agrees with that of her immediate left-hand neighbor---agent $i-1$ for $i\geq 2$, and the principal for $i=1$. 
The reward size depends on both the effort cost $c$ and the distance between the two neighboring task locations.

The next proposition characterizes the incentive cost  \eqref{inv}, namely, the infimum of expected compensations across all RIE contracts.
Also, the infimum can be attained in the limit by a sequence of contracts of the form $w^{\epsilon_m}$ defined in \eqref{wage}, where $\epsilon_m \to 0$.

\begin{proposition}\label{prop_inv}
For any positive sequence $\{\epsilon_m\}\rightarrow 0$, the sequence of RIE contracts  $\{w^{\epsilon_m}\}$ defined by \eqref{wage} approximates the incentive cost, i.e., 
\begin{equation}\label{inv_c}
K(\tau)=\lim_{m\rightarrow \infty} \hat K(w^{\epsilon_m},\tau)= 
\hat K (w^0, \tau)
=c\sum_{i=1}^n \left[1+e^{\lambda (t_i-t_{i-1})} \right]. 
\end{equation}
\end{proposition}

The key step in Proposition \ref{prop_inv} is to show that any RIE inducing a non-chain monitoring hierarchy can be improved. Once this is established, the optimal monitoring structure must be a chain, and the multi-agent contracting problem can then be characterized recursively, agent by agent. Then simple algebra shows that RIE contracts defined in \eqref{wage} are approximately optimal as $\epsilon_m\rightarrow 0$.

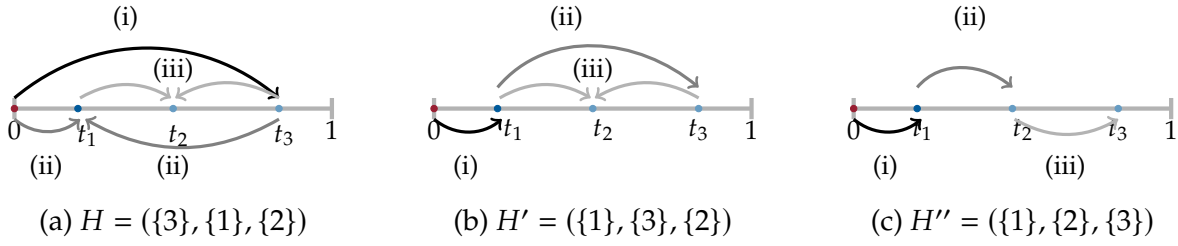
\begin{figure}[tbp]
\centering

% -------- Panel 1 --------
\begin{subfigure}[b]{0.32\textwidth}\label{fig:exa}
\centering
\begin{tikzpicture}[scale=1.4]
\footnotesize

\node at (0, 1.3) {$0$};
\node at (3, 1.3) {$1$};
\node at (0.69, 1.25) {$t_1$};
\node at (1.55, 1.25) {$t_2$};
\node at (2.55, 1.25) {$t_3$};

\draw [ultra thick, gray!60] (0,1.5) -- (3,1.5);
\draw [ultra thick, gray!60] (0,1.4) --(0,1.6);
\draw [ultra thick, gray!60] (3,1.4) --(3,1.6);

\draw [fill, pennred] (0,1.5) circle [radius=0.03];
\draw [fill, nberblue] (0.6,1.5) circle [radius=0.03];
\draw [fill, nberblue!60] (1.5,1.5) circle [radius=0.03];
\draw [fill, nberblue!60] (2.5,1.5) circle [radius=0.03];

\draw[very thick,->,gray] (0,1.4) to [out=-40,in=-140] (0.62,1.4);
\draw[very thick,->,gray] (2.5,1.4) to [out=-150,in=-30] (0.67,1.4);

\draw[very thick,->] (0,1.6) to [out=40,in=140] (2.5,1.6);

\draw[very thick,->, gray!60]  (0.62,1.6) to [out=30,in=150](1.48,1.6) ;
\draw[very thick,->, gray!60] (2.5,1.6) to [out=150,in=30] (1.52,1.6) ;

\node[above=6pt] at (1.05, 2) {(i)};

\node[below=6pt] at (0.3, 1.3) {(ii)};
\node[below=6pt] at (1.5, 1.3) {(ii)};

\node[above=6pt] at (1.5, 1.5) {(iii)};

\end{tikzpicture}
\caption{$H=(\{3\}, \{1\},\{2\})$}
\end{subfigure}
%\hfill
% -------- Panel 2 --------
\begin{subfigure}[b]{0.32\textwidth}
\centering
\begin{tikzpicture}[scale=1.4]
\footnotesize

\node at (0, 1.3) {$0$};
\node at (3, 1.3) {$1$};
\node at (0.75, 1.3) {$t_1$};
\node at (1.6, 1.3) {$t_2$};
\node at (2.5, 1.3) {$t_3$};

\draw [ultra thick, gray!60] (0,1.5) -- (3,1.5);
\draw [ultra thick, gray!60] (0,1.4) --(0,1.6);
\draw [ultra thick, gray!60] (3,1.4) --(3,1.6);

\draw [fill, pennred] (0,1.5) circle [radius=0.03];
\draw [fill, nberblue] (0.6,1.5) circle [radius=0.03];
\draw [fill, nberblue!60] (1.5,1.5) circle [radius=0.03];
\draw [fill, nberblue!60] (2.5,1.5) circle [radius=0.03];

\draw[very thick,->] (0,1.4) to [out=-40,in=-140] (0.64,1.4);

\draw[very thick,->, gray] (0.6,1.7) to [out=50,in=140] (2.5,1.7);

\draw[very thick,->, gray!60]  (0.62,1.6) to [out=30,in=150](1.48,1.6) ;
\draw[very thick,->, gray!60] (2.5,1.6) to [out=150,in=30] (1.52,1.6) ;

\node[above=6pt] at (1.25, 2) {(ii)};

\node[below=6pt] at (0.3, 1.3) {(i)};

\node[above=6pt] at (1.5, 1.5) {(iii)};

\end{tikzpicture}
\caption{$H'=(\{1\},\{3\},\{2\})$}
\end{subfigure}
%\hfill
% -------- Panel 3 (replicated) --------
\begin{subfigure}[b]{0.32\textwidth}
\centering
\begin{tikzpicture}[scale=1.4]
\footnotesize

\node at (0, 1.3) {$0$};
\node at (3, 1.3) {$1$};
\node at (0.65, 1.3) {$t_1$};
\node at (1.6, 1.3) {$t_2$};
\node at (2.5, 1.3) {$t_3$};

\draw [ultra thick, gray!60] (0,1.5) -- (3,1.5);
\draw [ultra thick, gray!60] (0,1.4) --(0,1.6);
\draw [ultra thick, gray!60] (3,1.4) --(3,1.6);

\draw [fill, pennred] (0,1.5) circle [radius=0.03];
\draw [fill, nberblue] (0.6,1.5) circle [radius=0.03];
\draw [fill, nberblue!60] (1.5,1.5) circle [radius=0.03];
\draw [fill, nberblue!60] (2.5,1.5) circle [radius=0.03];

\draw[very thick,->] (0,1.4) to [out=-40,in=-140] (0.64,1.4);
%\draw[very thick,->] (2.5,1.4) to [out=-150,in=-30] (0.67,1.4);

\draw[very thick,->, gray] (0.6,1.7) to [out=50,in=140] (1.5,1.7);

\draw[very thick,->, gray!60] (1.52,1.4)to [out=-30,in=-150] (2.5,1.4)  ;

\node[above=6pt] at (1.1, 2) {(ii)};

\node[below=6pt] at (0.3, 1.3) {(i)};

\node[below=6pt] at (2, 1.3) {(iii)};

\end{tikzpicture}
\caption{$H''=(\{1\},\{2\},\{3\})$}
\end{subfigure}

\caption{Illustration of Proposition 1}
\bigskip
{\scriptsize Notes: Roman numerals denote the order of elimination in the IESDS process.}
\label{fig_monitor_n}
\end{figure}

The proof proceeds by repeated local improvements. Starting from the top of any monitoring hierarchy, consider the first block in which monitoring authority is not passed downward in the order of task locations. Such a block must contain a local sandwich structure. Holding the rest of the hierarchy fixed, one can replace this sandwich by a local chain. By the two-agent argument, this replacement lowers the incentive cost of implementing the target effort profile. Iterating this step eliminates all local sandwiches and yields a chain hierarchy.

To illustrate the logic, consider the three-agent example in Figure~\ref{fig_monitor_n}. The hierarchy $H$ in panel (a) contains two nested sandwich structures. First, the principal monitors agent 3, and then the principal and agent 3 jointly monitor agent 1. Holding the rest of the hierarchy fixed, we can replace this local sandwich by a chain, obtaining the hierarchy $H'$ in panel (b). The hierarchy $H'$ still contains a sandwich: agent 1 monitors agent 3, and then agents 1 and 3 jointly monitor agent 2. Applying the same local replacement once more yields the chain hierarchy $H''$ in panel (c). By the two-agent argument, each replacement lowers the incentive cost. Hence the principal prefers the final chain hierarchy.

A corollary of Proposition \ref{prop_inv} is that the principal can reduce the incentive cost by assigning agents to locations that are closer to each other and to the principal:
\begin{corollary}\label{corollary1}
The incentive cost decreases when agents are closer to each other and closer to the principal.
Formally, take two task allocations, $\tau = (t_1,..., t_n)$ and $\tau' = (t'_1,..., t'_n)$, such that 
$t_i-t_{i-1} \le t'_i - t'_{i-1}$ for $i=1,..., n$ with $t_0 = t'_0=0$.
Then, we have $K(\tau) \le K(\tau')$, and $\min_{\tau}K(\tau)=K(0)\triangleq 2nc$.
\end{corollary}
Intuitively, under contract \eqref{wage}, agent $i$ is rewarded when her signal matches the signal of her immediate left-hand neighbor, whose incentive has already been secured in the chain. If agent $i$ works, closer task locations make this match event more likely. If instead agent $i$ shirks while her left-hand neighbor works, then her noisemaker-generated signal is unrelated to the neighbor's truthful signal, so the match event occurs with probability $1/2$. Thus, assigning agents to similar tasks increases the incremental bonus probability from effort, which enables the principal to lower the bonus and save on expected payment. 
In the extreme case where all agents are assigned to $t=0$, they are perfectly monitored by the principal. Even then, a shirking agent matches the principal’s signal with probability $1/2$. It follows that implementing effort requires an expected payment of $2c$ per agent. Of course, such an allocation is extremely inefficient in knowledge production, because the agents’ effort duplicates the principal’s existing information and creates no additional value.

\subsection{Optimal Contract under Bayesian Nash Implementation}\label{sec:BNI}

For comparison, we now set aside the robustness requirement and study the optimal contract under Bayesian Nash implementation.
Specifically, we require only that all agents exert effort in some Bayesian Nash equilibrium of the induced game. 
It allows the principal to rely on equilibrium coordination to resolve the strategic risk agents face about others’ actions.
This benchmark allows us to isolate how the robustness requirement shapes the implied monitoring structure.

For simplicity, we assume in this subsection that whenever agent $i$ shirks, the noisemaker draws $\shirk_i$ uniformly from $\{A,B\}$.\footnote{As shown in Appendix~\ref{appendix:bni-worst-noise}, for each shirking agent, uniform mixing is worst for the principal under Bayesian Nash implementation (BNI). This assumption therefore ensures that our conclusion---that robust implementation is more costly than BNI---is not an artifact of how the noisemaker's strategy is specified in the BNI benchmark.}
With the noisemaker's strategy fixed, we abuse notation and write $\Gamma(w,\tau)$ for the induced simultaneous-move Bayesian game among agents.

Given a task profile $\tau$, let $\contracts^{\mathrm{BNI}}(\tau)$ denote the set of contracts under which all agents exert effort in some Bayesian Nash equilibrium of $\Gamma(w,\tau)$.
A contract in $\contracts^{\mathrm{BNI}}(\tau)$ that minimizes the expected payment $\hat K(w,\tau)$ is called an \emph{optimal Bayesian Nash implementation (BNI) contract}.

The following result shows that the optimal BNI contract rewards each agent only if her signal coincides with the signals of both the left- and right-hand neighbors.
\begin{claim}\label{claim_bni}
An optimal BNI contract pays each agent $i\in\{1,\dots,n-1\}$ a bonus $b_i$ if and only if $\signal_{i-1}=\signal_i=\signal_{i+1}$, and pays agent $n$ a bonus $b_n$ if and only if $\signal_{n-1}=\signal_n$, where
\begin{equation}\label{bonus_bni}
b_i \;=\; \frac{4c}{e^{-\lambda(t_i-t_{i-1})}+e^{-\lambda(t_{i+1}-t_i)}}\quad(i\le n-1),
\qquad
b_n \;=\; 2c\,e^{\lambda(t_n-t_{n-1})}.
\end{equation}
\end{claim}

The intuition is as follows.
When her neighbors work, agent $i$'s neighboring signals $\signal_{i-1}$ and $\signal_{i+1}$ equal the true local states $\state(t_{i-1})$ and $\state(t_{i+1})$.
If agent $i$ also works, her signal $\signal_i$ equals $\state(t_i)$, which is more likely to coincide with both neighbors than when $\signal_i$ is a uniform draw independent of the state.
As a result, the likelihood ratio of effort versus shirking is maximized on the event $\signal_{i-1}=\signal_i=\signal_{i+1}$; the principal therefore minimizes the expected payment by rewarding only this event.
For agent $n$, no right-hand neighbor exists, so the bonus event collapses to $\signal_{n-1}=\signal_n$.

Because $\contracts(\tau)\subset \contracts^{\mathrm{BNI}}(\tau)$, the expected payment under the optimal BNI contract is lower than the RIE incentive cost.\footnote{More precisely, the total cost reduction is strict if and only if there exists some $i<n$ with $t_i>t_{i-1}$.}
Intuitively, the optimal BNI contract uses two neighboring signals to monitor each interior agent, whereas the RIE chain uses only the upstream one.
The extra signal increases the likelihood ratio of the bonus event under effort to that under no effort, which makes it cheaper for the principal to motivate effort.

The optimal robust contract and the optimal BNI contract have different monitoring structures.
Under the optimal BNI contract, agent $i$'s pay depends on agent $i+1$'s signal, and agent $i+1$'s pay depends on agent $i$'s signal.
Thus, the monitoring structure induced by the optimal BNI contract has cycles and is not hierarchical.
This is in contrast with the chain structure induced by the RIE optimum.

The following example illustrates how the cyclic monitoring structure induced by the optimal BNI contract can create an equilibrium in which no agent works.
Consider the two-agent case $n=2$.
The optimal BNI contract pays agent 1 a bonus $b_1=4c/(e^{-\lambda t_1}+e^{-\lambda(t_2-t_1)})$ when $\signal_0=\signal_1=\signal_2$, and pays agent 2 a bonus $b_2=2ce^{\lambda(t_2-t_1)}$ when $\signal_1=\signal_2$.
Under this contract, all agents shirking is an equilibrium.
To see this, suppose agent 1 shirks.
Her signal is then a uniform draw independent of the state, so the bonus event $\signal_1=\signal_2$ occurs with probability $1/2$ whether or not agent 2 works; consequently, agent 2 strictly prefers to shirk.
Suppose instead that agent 2 shirks.
Her signal is then uninformative about the state, which weakens the triple-match event and raises the minimum bonus required to support agent 1's effort above $b_1$; the optimal BNI bonus $b_1$ therefore fails to sustain agent 1's effort.\footnote{When agent 2 shirks, the triple match $\signal_0=\signal_1=\signal_2$ occurs with probability $\frac14(1+e^{-\lambda t_1})$ if agent 1 works and with probability $\frac14$ if agent 1 also shirks.
The binding incentive constraint for agent 1 therefore requires a bonus of $4ce^{\lambda t_1}$, which strictly exceeds $b_1$ because $e^{-\lambda t_1}+e^{-\lambda(t_2-t_1)}>e^{-\lambda t_1}$.}
As a result, the all-shirk profile is one of equilibria.

The Bayesian Nash implementation benchmark makes the price of robustness transparent.
The chain hierarchy of Proposition \ref{prop_inv} underuses information by ignoring the right-hand signal of every interior agent.
By doing so, it breaks the monitoring cycles that could otherwise induce coordinated shirking.

\section{Task Assignment}\label{sec_endo}
We return to robust implementation and now allow the principal to choose a task allocation.
In Section \ref{sec:first best}, we model how the principal uses the information generated by agents to improve his decisions.
Our specification captures the idea that when agents work on nearby tasks, their expertise overlaps, which lowers the marginal value of each additional signal.
This force favors task allocations that spread agents apart.
Combined with the incentive cost characterized in the previous section, the model yields a novel trade-off in knowledge production: assigning agents to more diversified tasks increases the value of the information they produce, but it also increases the incentive cost (cf.\ Corollary \ref{corollary1}).
Section \ref{sec:robust task} studies how the principal balances this trade-off under deterministic task allocation, and Section \ref{sec:ran} analyzes random and confidential task allocation.

\subsection{Value of Task Allocation}
\label{sec:first best}

We define the value of task allocation as the principal’s optimal use of the information generated by the team. To model how the principal uses this information, we suppose that the principal chooses an action $p(t)\in\{A,B\}$ at each location $t$ and incurs a unit cost if $p(t)\neq x(t)$. For each task allocation $\tau$, define \emph{information loss} as the ex ante minimized loss,

\[
L(\tau)=\mathbb{E}_{\signals\mid \tau}\left[
\min_{p_{\signals}:[0,1]\rightarrow\{A,B\}}
\mathbb{E}\left[
\int_0^1\mathbb{I}(p_{\signals}(t)\neq x(t))\,dt
\mid \tau,\signals
\right]\right].
\]
The following lemma describes the principal's optimal policy under any task allocation.
\begin{lemma}\label{lemma_policy}
Fix a task allocation, $t_1 \le \cdots \le t_n$.
Suppose all agents exert effort.
An action policy $p(\cdot)$ is optimal if and only if  $p(t) = \state(t^*)$ for almost every $t \in [0,1]$, where $t^*$ is a nearest sampled location for location $t$, i.e.,
    $t^* \in \argmin_{t' \in \{0,t_1,\dots,t_n\}} | t'-t|$. The information loss is
\begin{equation}
    L(\tau)= \frac{1}{2} - \frac{2n+1}{2\lambda} + \frac{1}{\lambda}\left[\sum_{i=1}^n e^{-\frac{\lambda (t_i - t_{i-1})}{2}} + \frac{1}{2}e^{-\lambda (1-t_n)}\right].\label{inf_loss}\end{equation}
\end{lemma}
The principal uses signals to infer whether each local state $\state(t)$ is more likely to be $A$ or $B$, then sets the action $p(t)$ to the more likely realization.
Because two local states are more likely to coincide if they are closer to each other, the optimal action $p(t)$ for location $t$ equals the local state of the nearest sampled location.\footnote{This property is in line with other papers that use spatial learning models to characterize the value of information; see, for example, \cite{Arjadanina} and \cite{bardhi2022attributes}.}
In doing so, the policy pointwise maximizes the probability that $p(t)=x(t)$, which is the principal's \emph{decision quality}, at each location. 

The following lemma characterizes the task allocation that minimizes information loss.
The principal would choose such a task allocation absent moral hazard.

\begin{lemma}[First-best Task Allocation]\label{prop_policy}
Information loss is minimized by the task allocation $\tau^{\dag} = (t_1^{\dag},..., t_n^{\dag})$ such that
\begin{equation}\label{assignment_fb}
t_i^{\dag}=\frac{2i}{2n+1}, \forall i=1,2,...,n.
\end{equation}
The information loss decreases in the number of agents $n$ and increases in the shock-arrival rate, $\lambda$.
\end{lemma}
The formula spreads agents across the task space: neighboring agents are equally spaced, and under the nearest-signal policy in Lemma \ref{lemma_policy}, each agent's signal determines the principal's action on an interval of the same length.

Agents should be spread out because assigning agents to nearby tasks generates information that improves decision quality over highly overlapping regions of the task space. 
Such overlap decreases the incremental value of each agent's signal.
Spreading agents across locations reduces this overlap, broadens coverage of the task space, and improves decision quality more evenly across the project.

\subsection{Task Allocation with Robust Incentives}\label{sec:robust task}

Once the robust incentive provision is taken into account, the 
optimal task allocation design minimizes the sum of information loss and incentive cost: 
\begin{equation}\label{eq:task_alloc_obj}
\min_{\tau} L(\tau)+K(\tau),
\end{equation}
where the information loss  $L(\tau)$ comes from the principal's optimal action policy in Lemma \ref{lemma_policy} and the incentive cost $K(\tau)$ in equation \eqref{inv_c}. 
The trade-off is rooted in the nature of knowledge production. If agents are assigned to similar tasks, they produce largely duplicated knowledge that creates little additional value for the principal.
However, the resulting overlap makes peer monitoring effective, which reduces the incentive cost.

The following result describes how robust implementation distorts task allocation.
Agents remain equally spaced, but the common spacing is smaller than in the first best, so the entire allocation is pulled closer to the principal (see \autoref{fig:robust_task_distortion}).
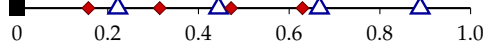
\begin{figure}[t]
\begin{center}
\begin{tikzpicture}[x=6cm, y=1cm,
    fbmarker/.style={regular polygon, regular polygon sides=3, draw=blue!60!black, fill=white, line width=0.9pt, minimum size=8pt, inner sep=0pt},
    sbmarker/.style={diamond, fill=red!80!black, draw=red!40!black, minimum size=5pt, inner sep=0pt},
    principal/.style={rectangle, fill=black, minimum size=6pt, inner sep=0pt}]
  \draw[thick] (0,0) -- (1,0);
  \foreach \x in {0, 0.2, 0.4, 0.6, 0.8, 1.0} {
    \draw (\x,0) -- (\x,-0.08);
    \node[below, font=\scriptsize] at (\x,-0.08) {\x};
  }
  \node[principal] at (0,0) {};
    \node[fbmarker] at (0.22222,0) {};
    \node[fbmarker] at (0.44444,0) {};
    \node[fbmarker] at (0.66667,0) {};
    \node[fbmarker] at (0.88889,0) {};
    \node[sbmarker] at (0.15733,0) {};
    \node[sbmarker] at (0.31466,0) {};
    \node[sbmarker] at (0.47200,0) {};
    \node[sbmarker] at (0.62933,0) {};
\end{tikzpicture}
\end{center}
\caption{Illustration of Distortion under Robust Implementation}
\bigskip

{\scriptsize Notes: First-best (blue triangles) versus second-best under robust implementation (red diamonds) for $n=4$, $\lambda=1$, and $c=0.1$. The principal (black square) is at $0$. The first-best allocation has common spacing $\Delta t^\dagger=2/9$, while the robust second-best allocation has common spacing $\Delta t^*\approx 0.157$.}
\label{fig:robust_task_distortion}
\end{figure}
\begin{proposition}\label{prop_task}
Relative to the first-best allocation, the optimal task allocation under robust implementation places agents closer to one another and closer to the principal. In particular, there exists a unique $\Delta t^* \in [0,\Delta t^\dagger)$ such that
$
t_i^*-t_{i-1}^*=\Delta t^*, \forall i=1,\dots,n,
$ 
with $t_0^*=0$, where
$
\Delta t^\dagger=\frac{2}{2n+1}
$ 
is the distance between any two neighboring agents under the first-best task allocation.
\end{proposition}

Proposition \ref{prop_task} shows that robust incentive provision systematically biases organizations toward familiar domains. It follows by writing the principal's objective in terms of the gaps
$\Delta_i=t_i-t_{i-1}$. 
Under the optimal RIE, the incentive cost $K(\tau)$ is additive across agents, and agent $i$'s contribution depends only on her distance $\Delta_i$ from her left-hand neighbor (see \eqref{inv_c}). 
The information loss \eqref{inf_loss} contains the term $\sum_{i=1}^n e^{-\lambda \Delta_i/2}$ and a terminal term $\frac12 e^{-\lambda(1-t_n)}$.
The sum exhibits a similar additive separability, while the terminal term  captures the remaining interval to the right of the last agent. 
Thus, conditional on the total reach $t_n=\sum_i \Delta_i$, the terminal term is fixed and the objective is a sum of identical convex functions of the gaps. 
As a result, it is optimal to equalize the neighboring distances. 
The principal then chooses the common $\Delta_i$ to balance the trade-off between stronger monitoring and broader informational coverage.

At the same time, this common spacing under the optimal robust contract is narrower than under the first best (see \autoref{fig:robust_task_distortion}).
In the first best, agents are spread out to avoid leaving any part of the task space too weakly covered. 
Under robust implementation, however, concentration also reduces the cost of incentive provision by placing agents in regions where monitoring is more effective and where the principal's own expertise is more useful. 
The organization therefore accepts poorer coverage in more distant regions of the task space in exchange for better-supervised knowledge production near the principal. 
In this sense, robust implementation leads the principal to prioritize dimensions of the project that are closer to her expertise and easier to monitor, while compromising on dimensions that are farther away and harder to discipline indirectly.

As a corollary of the above result, we obtain the following comparative statics: a larger organization can sustain high decision quality over a broader portion of the task space. 
\begin{corollary}\label{coro_task}
Suppose $c<(1-e^{-\lambda})/(2\lambda)$, so that $\Delta t^*>0$. As the number $n$ of agents increases, the optimal allocation becomes locally denser: the distance $\Delta t^*$ between each agent and her left-hand neighbor decreases, and so does the expected payment $c(1+e^{\lambda \Delta t^*})$ to each agent. At the same time, $t_n^*=n\Delta t^*$ increases.
\end{corollary}

As the number of agents increases, tasks are assigned more finely, so nearby agents work on more similar problems. This strengthens local monitoring and lowers the incentive cost for each agent. At the same time, because there are more agents, the organization can extend this finely monitored structure over a broader region of the task space. Thus, the organization’s high-quality core expands and the poorly covered fringe shrinks. Although robust implementation still favors concentration, greater organizational capacity mitigates its cost in terms of task-space coverage.

\begin{figure}[t]
\begin{center}
\begin{subfigure}[t]{0.48\textwidth}
  \centering
\begin{tikzpicture}[x=6cm, y=1cm,
    fbmarker/.style={regular polygon, regular polygon sides=3, draw=blue!60!black, fill=white, line width=0.9pt, minimum size=8pt, inner sep=0pt},
    sbmarker/.style={diamond, fill=red!80!black, draw=red!40!black, minimum size=5pt, inner sep=0pt},
    principal/.style={rectangle, fill=black, minimum size=6pt, inner sep=0pt}]
  \draw[thick] (0,0) -- (1,0);
  \foreach \x in {0, 0.2, 0.4, 0.6, 0.8, 1.0} {
    \draw (\x,0) -- (\x,-0.08);
    \node[below, font=\scriptsize] at (\x,-0.08) {\x};
  }
  \node[principal] at (0,0) {};
    \node[fbmarker] at (0.22222,0) {};
    \node[fbmarker] at (0.44444,0) {};
    \node[fbmarker] at (0.66667,0) {};
    \node[fbmarker] at (0.88889,0) {};
    \node[sbmarker] at (0.24777,0) {};
    \node[sbmarker] at (0.48402,0) {};
    \node[sbmarker] at (0.72521,0) {};
    \node[sbmarker] at (0.87070,0) {};
\end{tikzpicture}
  %\caption{$c = 0.02$: agents $1,2,3$ shift right, agent $4$ shifts left.}
  \label{fig:agent_locations_n4_small}
\end{subfigure}\hfill
\begin{subfigure}[t]{0.48\textwidth}
  \centering
\begin{tikzpicture}[x=6cm, y=1cm,
    fbmarker/.style={regular polygon, regular polygon sides=3, draw=blue!60!black, fill=white, line width=0.9pt, minimum size=8pt, inner sep=0pt},
    sbmarker/.style={diamond, fill=red!80!black, draw=red!40!black, minimum size=5pt, inner sep=0pt},
    principal/.style={rectangle, fill=black, minimum size=6pt, inner sep=0pt}]
  \draw[thick] (0,0) -- (1,0);
  \foreach \x in {0, 0.2, 0.4, 0.6, 0.8, 1.0} {
    \draw (\x,0) -- (\x,-0.08);
    \node[below, font=\scriptsize] at (\x,-0.08) {\x};
  }
  \node[principal] at (0,0) {};
    \node[fbmarker] at (0.22222,0) {};
    \node[fbmarker] at (0.44444,0) {};
    \node[fbmarker] at (0.66667,0) {};
    \node[fbmarker] at (0.88889,0) {};
    \node[sbmarker] at (0.38800,0) {};
    \node[sbmarker] at (0.41200,0) {};
    \node[sbmarker] at (0.78800,0) {};
    \node[sbmarker] at (0.81200,0) {};
\end{tikzpicture}
  \label{fig:agent_locations_n4_large}
\end{subfigure}
\end{center}
\caption{Illustration of Distortion under Bayesian Nash Implementation}
\bigskip

{\scriptsize Notes: First-best (blue triangles) versus second-best under Bayesian Nash implementation (red diamonds) for $n=4$, $\lambda=1$. The principal (black square) is at $0$. In the left panel, $c = 0.02$: agents $1,2,3$ shift right, agent $4$ shifts left. In the right panel, $c = 0.5$: agents $(1,2)$ pair near $\tfrac{2}{5}$; agents $(3,4)$ pair near $\tfrac{4}{5}$.}
\label{fig:agent_locations_n4}
\end{figure}
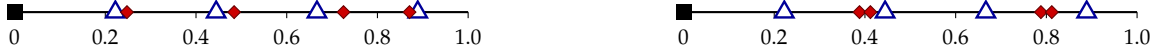

\begin{remark}[Task allocation under Bayesian Nash implementation.]  The distortion in \autoref{prop_task}, whereby the principal locates agents closer to his own location, is a consequence of the robustness requirement rather than a generic implication of moral hazard. To see this, consider instead the optimal task allocation under the Bayesian Nash implementation (BNI) benchmark of Section \ref{sec:BNI}, which minimizes the sum of the information loss \eqref{inf_loss} and the minimized payment under BNI in \hyperref[claim_bni]{Claim \ref{claim_bni}}.
\autoref{fig:agent_locations_n4} compares the first-best allocation with the BNI second-best allocation for four agents at two effort-cost levels.
The left panel shows that when the effort cost is small, the direction of distortion under BNI can be the opposite of that in \autoref{prop_task}: all interior agents move \emph{away} from the principal relative to the first best.
The right panel shows that a qualitative contrast between the BNI optimum and \autoref{prop_task} persists even when the effort cost is large.\footnote{For an even number of agents and a sufficiently large effort cost $c$, the BNI-optimal allocation groups agents into $n/2$ pairs and places each pair at the same location, where the chosen locations coincide with the first-best locations for $n/2$ agents.
In contrast, under robust implementation all agents collapse to $t=0$ for large $c$.}
\end{remark}

\subsection{Confidential Assignment and Perceived Monitoring Hierarchies}\label{sec:ran}

Now we investigate the value of stochastic confidential task allocation, under which each agent observes only her own realized assignment but not those of others. The key is to introduce uncertainty about the monitoring structure: an agent no longer knows exactly who is assigned to nearby tasks, and hence cannot precisely assess how likely her effort is to be detected through peer evaluation. We ask whether, and to what extent, this confidentiality in task allocation helps the principal robustly implement effort.

Confidential assignment is natural in many modern knowledge-production environments. On crowdsourcing platforms such as Amazon Mechanical Turk, for example, a client can assign the same task to multiple workers without revealing how many others receive the same assignment. More broadly, in remote work environments, team members often have limited knowledge of one another’s assignments and may not know who else is working on closely related tasks. Similar features also arise when firms hire external consultants, when faculty hire research assistants from a large pool of students who do not know one another, or when journals and funding agencies rely on anonymous reviewers drawn from the same pool of researchers or applicants whose own work is also being evaluated.

Formally, we modify the model as follows:
First, the principal publicly commits to a stochastic task allocation, which is a pair $(T, \pi)$ such that $T\subseteq [0,1]^n$ is a finite set of deterministic task allocations and $\pi\in \Delta(T)$ is a probability distribution over $T$.
The principal also commits to a contract, which specifies, for each realized task allocation $\tau \in T$ and a signal profile $\signals$, the payment $w_i(\tau, \signals) \ge 0$ for each agent $i$.
Abusing notation, we continue using $w$ for a contract of this model.

Once the principal chooses $(T, \pi, w)$,  Nature draws the state $x(\cdot)$ and the task allocation, $\tau  = (t_1,..., t_n)\sim \pi$.
Each agent $i$ observes $(T, \pi, w)$ and $t_i$, and then decides whether to acquire information.
Agents do not observe the realizations of other agents' locations.
For example, if the principal randomizes over $(t_1 , t_2) = (0.1, 0.3)$, $(0.1, 0.4)$, and $(0.2,0.5)$ with equal probability, then agent $1$ observes only whether $t_1 =0.1$ or $t_1 = 0.2$, and after observing $t_1 = 0.1$, she infers that agent $2$'s location is $0.3$ or $0.4$ with equal probability.

Our notion of robust implementation extends to this setup by applying IESDS to the agent-normal form of the incomplete-information game $\Gamma(T,\pi,w)$.
For each private type $t_i$ of agent $i$, the corresponding type-agent chooses whether to acquire information.
As before, the noisemaker is payoff-irrelevant; here, its pure strategy is an allocation-contingent noise rule $\eta:T\to\{A,B\}^n$, so that when allocation $\tau$ is realized and agent $i$ shirks, her signal is $\eta_i(\tau)$.
Contract $w$ robustly implements effort if every type-agent's shirking action is eliminated under IESDS.

The following result shows that confidential randomization can asymptotically achieve the first-best information loss while reducing the incentive cost to $K(0)$.
\begin{proposition}\label{random_eff2}
Suppose that there are at least two agents.
There is a sequence of stochastic task allocations and contracts, $(T_k, \pi_k, w_k)_{k \in \mathbb{N}}$, such that for each $k$, contract $w_k$ robustly implements effort, and as $k\to\infty$, the principal's total loss converges to $L(\tau^{\dag})+K(0)$.
\end{proposition}

 \autoref{random_eff2} shows that stochastic task assignment, combined with confidentiality of both assignments and payments across agents, can mitigate the tension between diversification benefits and incentive costs. In particular, the principal can simultaneously approximate (i) the first-best information loss and (ii) the minimal second-best expected payments as if all agents were assigned to location $0$. The random assignment creates uncertainty and asymmetric perceptions across agents about the realized monitoring hierarchy. The principal can exploit these uncertainty to robustly implement effort at lower cost. This virtual implementation outcome requires contracts that deliver arbitrarily large payments with arbitrarily small probabilities.

Next, we sketch the proof for the case of two agents and leave the general proof to Appendix \ref{appendix:random_task}. Consider the following choice of the principal:
With probability $1-p$, the principal chooses the first-best task allocation $\tau^\dag = (t^\dag_1, t^\dag_2)$, which is  $(\frac25,  \frac45)$ by  \autoref{prop_policy}.
For the remaining event, the principal assigns both agents to the same location, i.e., $\tau^\dag_1 \triangleq (t^\dag_1,t^\dag_1)$ or $\tau^\dag_2 \triangleq (t^\dag_2, t^\dag_2)$ with probability $0.5p$.
Task allocation $\tau^\dag_i$ assigns both agents to agent $i$'s first-best location, $t^\dag_i$.
The contract is as follows:\
If the realized task allocation is the first-best $\tau^\dag$, there is no payment.
Under task allocation $\tau^\dag_i$, agent $j\not=i$ will receive $2c e^{\lambda t^\dag_i}+\epsilon$ if her signal matches the principal's signal, $x(0)$, and
agent $i$ will receive $\frac{4c}{p}+\epsilon$ if her signal matches agent $j$'s signal.

The key property of this policy is that even when the realized task allocation is the first-best, agent $i$ still believes that the task allocation can be $\tau^\dag_i$ and the other agent is assigned to the same location as agent $i$.
While the probability of $\tau^\dag_i$ can be small (as we will take $p\to 0$), the principal sets the bonus $\frac{4c}{p}+\epsilon$ for agent $i$ large enough so that she finds it optimal to acquire information.

Indeed, for any $p \in (0,1)$ and $\epsilon>0$, the IESDS uniquely selects both agents acquiring information:
Suppose that agent $1$ learns that her assigned location is agent $2$'s first-best location $t^\dag_2$.
Because the possible task allocations are $\tau^\dag =(t^\dag_1, t^\dag_2)$, $\tau^\dag_1= (t^\dag_1, t^\dag_1)$, and $\tau^\dag_2= (t^\dag_2, t^\dag_2)$,
agent $1$ learns that the realized task allocation must be $\tau^\dag_2$.
As a result, agent $1$ understands that she will receive $2c e^{\lambda t^\dag_2}+\epsilon$ if her signal matches the principal's signal, $x(0)$, regardless of the other agent's action. Thus, 
agent $1$ strictly prefers to exert effort as a strictly dominant strategy.
Symmetrically, agent $2$ will work when the realized task allocation is $\tau^\dag_1$.

Second, suppose that agent $1$ learns that she is assigned to her first-best location $t^\dag_1$.
Agent $1$ then infers that the task allocation is $\tau^\dag$ or $\tau^\dag_1$ with probability $\frac{1 - p}{1-p+0.5p}$ or $ \frac{0.5p}{1-p+0.5p}$, respectively.
Agent $1$ could possibly earn a positive payment only when the realization allocation $\tau^\dag_1$, but agent $1$ knows that, in such a case, agent $2$ acquires information (because of the first step of the IESDS).
Thus, by exerting effort, agent $1$ can increase the probability of the positive payment $\frac{4c}{p}+\epsilon$ by $0.5$ in case the task allocation $\tau^\dag_1$ is realized.
As a result, the expected net gain for agent $1$ of acquiring information after observing $t^\dag_1$ is
\begin{equation*}
\frac{0.5p}{1-p+0.5p} \cdot 0.5 \cdot \left(\frac{4c}{p}+\epsilon\right) -c >
0.5p \cdot 0.5 \cdot \left(\frac{4c}{p}+\epsilon\right) -c =\frac{p\epsilon}{4}>0.
\end{equation*}
Thus, agent $1$ exerts effort as a strictly dominant strategy in the second step of the IESDS.
By applying the symmetric argument to agent $2$, we conclude that each agent $2$ also chooses the desired strategy as a dominant strategy.
Therefore, the IESDS selects the desired outcome for any realized task allocation.

Finally, by taking $(p, \epsilon) \to 0$, the principal will implement task allocation $\tau^\dag$ with probability converging to $1$, so the information loss converges to the first-best level.
At the same time, the expected payment to each agent converges to $0.5p\cdot \frac{4c}{p} = 2c$.
Here, the payment regarding $2c e^{\lambda t^\dag_i}+\epsilon$ will vanish as $p\to 0$, because it is independent of $p$ and can arise only with probability $0.5p$.
Therefore, the principal's total loss can be arbitrarily close to $L(\tau^\dag)+ K(0)$.\footnote{Our construction of stochastic assignment is related in spirit to \cite{legros1993efficient}, \cite{rahman2010mediated}, and \cite{rahman2012but}, which study contracts that occasionally induce some agents to take ex post suboptimal actions in order to help identify non-deviators. Our setting differs in two respects. First, we study robust implementation. Second, the principal creates a random monitoring structure through task assignment, rather than through correlated recommendations in a fixed game.}

\begin{remark}[Applicability.]  
Practical relevance of random assignment depends on the application. It may be limited by the principal's liquidity constraints, imperfect confidentiality of task assignments, infeasibility of assigning different agents to exactly the same task, and agents' risk aversion. For example, assigning multiple agents to the same task may be difficult when agents must share equipment, use common accounts, work in the same laboratory, or coordinate on a common schedule. These features can make assignments observable and undermine the confidentiality needed for the construction. Even so, the broader lesson is general: making agents uncertain about others' assignments can significantly relax the tension between diversity and incentive and thereby improve the principal's objective.
\end{remark}

\section{Conclusion}\label{sec_con}
We study peer monitoring design in knowledge production. 
Agents' work generates information useful for evaluating related tasks, 
but this information can support incentives only after the agents who produce it have been induced to work.
We show that robust implementation imposes an ordering on monitoring authority, and in particular, the optimal contract organizes peer monitoring as a chain ordered by proximity to the principal's core expertise. The principal first disciplines the agent whose task is easiest for her to monitor; that agent's signal then becomes a credible monitor for the next agent, and so on.
This one-directional hierarchy breaks circular incentive dependence, but it also limits the use of information that would be available under mutual monitoring.
When task allocation is endogenous, the hierarchy distorts knowledge production toward familiar, easily monitored tasks. Confidential random assignment can relax this distortion by making agents uncertain about the realized monitoring hierarchy.
%The analysis highlights a distinct organizational role of task assignment: it shapes not only what information is produced, but also who can credibly evaluate whom.

Our analysis abstracts from collusion among agents. 
The absence of collusion is plausible in some knowledge-production environments, such as when agents work remotely, interact only through formal reporting channels, or when parts of the project are outsourced to agents with limited direct communication. Extending the model to allow for collusion is a natural direction for future work. Our results suggest that the implications may depend on what forms of collusion are feasible. If agents can communicate with each other but cannot make side payments, simple modifications of the contract may deter information sharing; for example, the principal could penalize evaluation patterns that are too closely aligned with the information of the evaluated agent. If agents can make side payments, the problem is more serious: an agent may bribe her evaluator to preserve her own compensation. In the spirit of \cite{strulovici2021learning}, collusion-proofing might require giving agents higher in the monitoring hierarchy sufficiently strong incentives so that the largest bribe offered by those they monitor cannot overturn their willingness to evaluate them.

\bibliographystyle{apalike}
\bibliography{reference}

\appendix

\section{Appendix: Omitted Proofs and Technical Details}\label{sectionAppendix}

\subsection{Proof of \autoref{prop-RIE}} \label{apx:proof-RIE}

   The proof of the first part of \autoref{prop-RIE} is provided in the main text. In this section, we prove the second part, i.e., for any RIE contract, there exists an RIE contract that induces a monitoring hierarchy and attains weakly lower incentive cost.
   Given a fixed task profile $\tau$ such that $t_1\le t_2\le \cdots \le t_n$.  Take any RIE contract. 
Without loss, assume that in each step of its iterated elimination of strictly dominated strategies (IESDS), exactly one agent eliminates shirking.
We can identify the order of elimination with a permutation of agents' names, generically written as $\permutation = (i_1,..., i_n)$, i.e., agent $i_k$ is the $k$-th agent to eliminate shirking.
We call such permutation $\permutation$ an \emph{elimination order}.
Given any elimination order $\permutation$, we say that an agent $i$ is \emph{$k^{th}$-ranked} if $i = i_k$.
Agent $i_j$ is \emph{risk-free for agent $i_k$} 
if $j < k$, i.e., agent $i_j$ eliminates shirking earlier than agent $i_k$ according to order $\permutation$.
Throughout this proof, we fix an arbitrary elimination order $\permutation$ induced by the RIE contract.

The proof consists of three steps.
{\bf Step 1} identifies a lower bound of expected payments for each agent $i_k$ ($k=1,..., n$) across all RIE contracts consistent with elimination order $\permutation$.
By adding up these lower bounds across agents, we obtain a lower bound of total expected payments across all RIE contracts consistent with elimination order $\permutation$.
{\bf  Step 2} constructs a family of contracts, $\{w^{* \epsilon}\}$, with respect to $\permutation$.
We show that, as $\epsilon \to 0$,  $w^{* \epsilon}$ attains the lower bound in Step 1.
Finally, {\bf Step 3} shows that for any $\epsilon > 0$, contract $w^{* \epsilon}$ is an RIE contract, and it induces a monitoring hierarchy consistent with $\permutation$.

\medskip \noindent\textbf{Step 1: Lower bound of expected payment for a given elimination order.} 
Unless otherwise stated, ``lower bound for agent $i_k$'' means a lower bound of expected payments for agent $i_k$ across all RIE contracts consistent with elimination order $\permutation$, where the expected payments are evaluated at all agents exerting effort.

This step consists of four substeps:
{\bf Step 1.1} derives a lower bound for agent $i_1$.
{\bf Step 1.2} restricts our attention to a class of contracts to derive a lower bound for other agents.
{\bf Step 1.3} derives a lower bound for agent $i_2$, and {\bf Step 1.4} derives those for the rest of agents.

  \medskip \noindent\textbf{Step 1.1: Lower bound of payment for agent $i_1$.} 
  Let $\contracts(\tau, \permutation)$ denote the set of all RIE contracts that admit elimination order $\permutation$.
  In the game $\Gamma(\tau, w)$ with any contract $w\in \contracts(\tau, \permutation)$, agent $i_1$ strictly prefers to work regardless of the strategies of other agents and the noisemaker.

    Specifically, because agent $i_1$ is the first to eliminate shirking, any other agent $j\not= i_1$ has both working $\action_j = 1$ and shirking $\action_j = 0$ uneliminated. 
    For any agent $j$ who works, her signal is the local state of her assigned location, i.e., $\signal_j = \state(t_j)$.
    For any agent $j \neq i_1$ who shirks, the noisemaker's choice, $\shirks = (\shirk_1,..., \shirk_n)$, determines her signal realization as $\signal_j  = \shirk_j\in \{A,B\}$. 
    
    Given any combination of other agents' strategies $\actions_{-i_1} \in \actionspace_{-i_1}^* \triangleq \{0,1\}^{n-1}$ and the noisemaker's choice $\shirks \in \{A, B\}^n$, any contract $\contract \in \contracts(\tau, \permutation)$ must be such that agent $i_1$ earns a strictly higher payoff by working than shirking.
  Thus, $\forall \actions_{-i_1} \in  \actionspace_{-i_1}^{*}, \forall \shirks \in \{A,B\}^{n},$ the contract must satisfy the following set of inequalities:
    \begin{equation}\label{equationi1}
        \sum_\signals w_{i_1}(\signals) \Pr(\signals|\action_{i_1} = 1, \actions_{-i_1},\shirks) -c > \sum_\signals w_{i_1}(\signals)\Pr(\signals|\action_{i_1} = 0, \actions_{-i_1}, \shirks), 
    \end{equation}
    where $\Pr(\signals|\action_{i_1} , \actions_{-i_1},\shirks)$ is the probability of signal profile $\signals$ given agents' strategies and the noisemaker's choice.

The goal of this substep is to find a lower bound of the expected payment for agent $i_1$ across all contracts that satisfy \eqref{equationi1}, where the expected payment is evaluated at all agents working.
(Hereafter, whenever we say ``payment'' it means the payment evaluated at all agents working.) Specifically, we find a lower bound by minimizing the expected payment for agent $i_1$ subject to weaker constraints that modify \eqref{equationi1} as follows: (i) imposing only the constraints associated with all other agents shirking, i.e., $\actions_{-i_1} = \boldsymbol{0}$; and (ii) replacing strict inequalities with weak ones.
Meanwhile, it proves convenient to write the noisemaker's choice separately for agent $i_1$  as $\shirk_{i_1}   \in \{A, B\}$ and for other agents as $\shirks_{-i_1} \in \{A, B\}^{n-1}$.

The weaker constraints that reflect Points (i) and (ii) above are written as:
\begin{align}
        &\sum_\signals w_{i_1}(\signals) \Pr(\signals|\action_{i_1} = 1, \actions_{-i_1}= \boldsymbol{0},\shirk_{i_1} = A, \shirks_{-i_1}) -c \geq \sum_\signals w_{i_1}(\signals)\Pr(\signals|\action_{i_1} = 0,\actions_{-i_1}= \boldsymbol{0},\shirk_{i_1} = A, \shirks_{-i_1})\label{equationICA}\\
        &\sum_\signals w_{i_1}(\signals) \Pr(\signals|\action_{i_1} = 1, \actions_{-i_1}= \boldsymbol{0},\shirk_{i_1} = B, \shirks_{-i_1}) -c \geq \sum_\signals w_{i_1}(\signals)\Pr(\signals|\action_{i_1} = 0,\actions_{-i_1}= \boldsymbol{0},\shirk_{i_1} = B, \shirks_{-i_1})\label{equationICB}
\end{align}
  for all $\shirks_{-i_1} \in \{A,B\}^{n-1}$. 
In \hyperref[apx:omitted_details]{Appendix \ref{apx:omitted_details}}, we show that, 
for any fixed $\shirks_{-i_1} \in \{A, B\}^{n-1}$, constraints \eqref{equationICA} and \eqref{equationICB} are equivalent to, respectively, 
 \begin{align}
 &\frac{1}{4}(1-e^{-\lambda t_{i_1}})[w_{i_1}(AB) - w_{i_1}(AA)] + \frac{1}{4}(1+e^{-\lambda t_{i_1}})[w_{i_1}(BB) - w_{i_1}(BA)] \geq c \tag{ICA$_{\shirks_{-i_1}}$};\quad  \text{and} \label{ICA1.1}\\
 &\frac{1}{4}(1-e^{-\lambda t_{i_1}})[w_{i_1}(BA) - w_{i_1}(BB)] + \frac{1}{4}(1+e^{-\lambda t_{i_1}})[w_{i_1}(AA) - w_{i_1}(AB)] \geq c\tag{ICB$_{\shirks_{-i_1}}$} \label{ICB1.1},
 \end{align}
where we adopt the shorthand notation $w_{i_1}(\signal_0\signal_{i_1}):=w_{i_1}(\signal_0\signal_{i_1}; \shirks_{-i_1})$ for each $\signal_0\signal_{i_1} \in \{AA, AB, BA, BB\}$. 

We show that at the optimum---i.e., at the contract that minimizes the payment for agent $i_1$ given these constraints---the constraints must be binding.
Suppose, for example, that ICA$_{\shirks_{-i_1}}$ is slack, then we have $w _{i_1}(AB)>0$ or $w_{i_1}(BB)>0$. 
If, e.g., $w_{i_1}(AB) > 0$, then the principal can decrease $w_{i_1}(AB)$ by some small $\epsilon>0$, such that ICA$_{\shirks_{-i_1}}$ is still slack, ICB$_{\shirks_{-i_1}}$ is relaxed, and the expected payment for $i_1$ strictly decreases.
 A similar argument works for the case in which ICB$_{\shirks_{-i_1}}$ is slack. 
 
 Given that constraints ICA$_{\shirks_{-i_1}}$ and ICB$_{\shirks_{-i_1}}$ bind at the optimum, the system of equations uniquely solves $w_{i_1}(AA)  - w_{i_1}(AB)= w_{i_1}(BB) - w_{i_1}(BA) = 2ce^{\lambda t_{i_1}}$. 
 The principal should then reduce $w_{i_1}(AB)$ and $w_{i_1}(BA)$ as much as possible, until 
 the limited liability constraint binds, i.e., $w_{i_1}(AB) = w_{i_1}(BA) = 0$.
 We then obtain $w_{i_1}(AA) = w_{i_1}(BB) = 2ce^{\lambda t_{i_1}}$.

 The procedure works for each $\shirks_{-i_1}$.
 Note that although $\shirks_{-i_1}$ denotes the noisemaker's choice, the resulting payment rule, such as $\contract_{i_1}(\signal_0\signal_{i_1}; \shirks_{-i_1})$, pins down the payment to agent $i_1$ when the realized signals for other agents are $\signals_{-i_1} = \shirks_{-i_1} \in \{A, B\}^{n-1}$.
 
 The resulting contract takes the following simple form: Regardless of the realized signals of other agents, agent $i_1$ earns $2ce^{\lambda t_{i_1}}$ if $\signal_0 = \signal_{i_1}$ and $0$ otherwise.
Formally, we obtain 
$w_{i_1}^*(AA;\signals_{-i_1}) = w_{i_1}^*(BB;\signals_{-i_1}) = 2ce^{\lambda t_{i_1}}$ and $w_{i_1}^*(AB;\signals_{-i_1}) = w_{i_1}^*(BA;\signals_{-i_1}) = 0$ for all $\signals_{-i_1} \in \{A,B\}^{n-1}$.
The corresponding lower bound of payment for agent $i_1$ is \[
\hat{K}_{i_1}^\permutation(\tau,w^*) = \sum_{\signals\in \{A,B\}^{n+1}}w_{i_1}^*(\signals)\Pr(\signals|\tau, \actions = \mathbf{1}) = 2ce^{\lambda t_{i_1}}\cdot\Pr(x_0 = x_{i_1}) = c(1+e^{\lambda t_{i_1}}).
\]

{

\medskip \noindent\textbf{Step 1.2: Simplifying contracts.} 
To characterize lower bounds for other agents, 
we first simplify the class of contracts we need to focus on.
Recall that for a given $k \ge 2$, any agent $i_j$ with $j< k$ is risk-free for agent $i_k$ in that agent $i_j$ eliminated shirking in an earlier step of the IESDS than agent $i_k$.
Any agent $i_j$ with $j>k$ is risky in that she has not eliminated shirking in the $k$-th round.
Hereafter, we use ``risky'' and ``risk-free'' from the perspective of agent $i_k$.

Fix agent $i_k$, we partition the set of other agents as follows. Select the \textit{left-risk-free-neighbor} by \[
L := \max(\{i_j\colon j<k, \nexists i_\ell \text{ such that } \ell < k \text{ and } t_{i_j} < t_{i_\ell} < t_{i_k}\} \cup \{0\}),
\]
i.e., the risk-free agent who is (i) located on the left side of agent $i_k$ and (ii) closer to agent $i_k$ than any other risk-free agents located on agent $i_k$'s left side.\footnote{
We use the ``left side of agent $i_k$'' to mean any location $t\le t_{i_k}$.
Similarly, the ``right side of agent $i_k$'' refers to any location $t \ge t_{i_k}$.
} If no agent satisfies the requirements of (i) (ii), select the principal as the left-risk-free-neighbor (i.e., $L = 0$) and set $t_L = 0$.

Symmetrically, select the \textit{right-risk-free-neighbor} as \[
R := \min(\{i_j\colon j<k, \nexists i_\ell \text{ such that } \ell < k \text{ and }t_{i_k} < t_{i_\ell} < t_{i_j}\}\cup \{+\infty\}),
\]
i.e., search for the risk-free agent who is (i) located on the right side of agent $i_k$, (ii) closer to agent $i_k$ than any other risk-free agents located on agent $i_k$'s right side. If no agent satisfies the requirements of (i) (ii), set $R = +\infty$ and $t_{R} = +\infty$, i.e., a hypothetical location whose state is uncorrelated with that of any location $t \leq 1$.

Denote the set of other risk-free agents (or the principal, if applicable) as \[
M := \{i_j\colon j < k\} \cap (\mathcal{N}\cup \{0, +\infty\} \setminus\{L, R\}).
\]
Finally, define two sets of risky agents as \[
N := \{i_j\colon j > k \text{ and } t_L< t_{i_j}<t_R\}, O := \{i_j\colon j > k \text{ and } t_{i_j} \leq t_L \text{ or } t_{i_j}\geq t_R\}, 
\]
i.e., the \textit{inner} risky agents who are located closer to $i_k$ than $L$ or $R$, and the \textit{outer} risky agents who are not inner, respectively. 

From this step onward, we find a lower bound of the expected payment for agent $i_k$ by (i) imposing only the constraints associated with all \textit{outer} risky agents and risk-free agents working, while \textit{inner} risky agents shirking; and (ii) replacing strict constraints into weak ones. We show that upon finding a lower bound for agent $i_k$,
we can focus on contracts that do not use the signals of risk-free non-neighbors in $M$ or \emph{outer} risky agents in $O$.
Recall that $\action_{i_k} = 1$ denotes agent $i_k$'s choice to work.
With a slight abuse of notation, we write $a_{i_k} = A$ or $a_{i_k} =B$ for agent $i_k$'s shirking and the noisemaker's choosing $\shirk_{i_k} = A$ or $\shirk_{i_k} = B$, respectively. Denote $\signals_M$, $\signals_N$, $\signals_O$ as the signal profile for all agents in the set $M$, $N$, $O$, respectively.

Take any contract for agent $i_k$, $w_{i_k}$.
Then, for any signal profile, $(\signal_{i_k}, \signal_{L}, \signal_{R}, \signals_{M}, \signals_{N}, \signals_{O})$, 
define $\hat{\contract}_{i_k}$ as 
\[
\hat{\contract}_{i_k}(\signal_{i_k}, \signal_{L}, \signal_{R}, \signals_{M},  \signals_{O}; \signals_{N}) 
\triangleq 
\mathbb{E}_{\tilde{\signals}_{M}, \tilde{\signals}_{O}}[\contract_{i_k}(\signal_{i_k}, \signal_{L}, \signal_{R}, \tilde{\signals}_{M},  \tilde{\signals}_{O}; \signals_{N}) |\signal_{L}, \signal_{R}].
\] 
That is, we construct $\hat{\contract}_{i_k}$ by taking expectation with respect to $(\signals_{M}, \signals_{O})$ conditional on $(\signal_{L}, \signal_{R})$.
The resulting $\hat{w}_{i_k}$ is independent of the signals $(\signals_{M}, \signals_{O})$ of agent $i_k$'s risk-free non-neighbors and outer risky agents. Recall that we are considering ICs in which (i) agent $i_k$ chooses action $\action_{i_k}\in \{1, A, B\}$; (ii) all \textit{outer} risky agents and risk-free agents  are working; and (iii) all \textit{inner} risky agents are shirking and suppose their realized signals (i.e., the noisemaker's choices) are $\signals_{N}\in \{A,B\}^{|N|}$ with probability $1$.
Then, the expected payments for agent $i_k$ under these two contracts are equal and the ICs satisfying (i)(ii)(iii) are preserved,  because 
\begin{align*}
    \mathbb{E}[w_{i_k}(\signal_{i_k}, \signal_{L}, \signal_{R}, \signals_{M}, \signals_{O}; \signals_{N})|\action_{i_k}] 
    &= \mathbb{E}[\mathbb{E}_{\tilde{\signals}_{M}, \tilde{\signals}_{O}}[w_{i_k}(\signal_{i_k}, \signal_{L}, \signal_{R}, \tilde{\signals}_{M},  \tilde{\signals}_{O}; \signals_{N}) |\action_{i_k}, \signal_{L}, \signal_{R}]|\action_{i_k}]\\
    &= \mathbb{E}[\mathbb{E}_{\tilde{\signals}_{M}, \tilde{\signals}_{O}}[w_{i_k}(\signal_{i_k}, \signal_{L}, \signal_{R}, \tilde{\signals}_{M},  \tilde{\signals}_{O}; \signals_{N})|\signal_{L}, \signal_{R}]|\action_{i_k}]\\
        &=\mathbb{E}[\hat{\contract}_{i_k}(\signal_{i_k}, \signal_{L}, \signal_{R}, \signals_{M},  \signals_{O}; \signals_{N}) |\action_{i_k}],
\end{align*}
for each $ \action_{i_k} \in \{1,A,B\},  \signals_{N}\in \{A,B\}^{|N|}$, and $(\signal_L, \signal_R) \in \{A,B\}^2$; the first equality is the law of iterated expectation and the second equality follows from the Markov property.
Therefore, it is without loss to focus on contracts that do not use risk-free non-neighbors' signals $\signals_{M}$ or outer risky signals $\signals_{O}$. Figure \ref{fig:simple-contract} illustrates how the contract is simplified after Step 1.2.

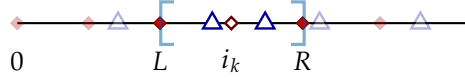
\begin{figure}[t]
\begin{center}
\begin{tikzpicture}[x=6cm, y=1cm,
    fbmarker/.style={
        regular polygon,
        regular polygon sides=3,
        draw=blue!60!black,
        fill=white,
        line width=0.9pt,
        minimum size=8pt,
        inner sep=0pt
    },
    fbmarkerfade/.style={
        fbmarker,
        draw opacity=0.3,
        fill opacity=0.3
    },
    sbmarker/.style={
        diamond,
        fill=red!80!black,
        draw=red!40!black,
        minimum size=5pt,
        inner sep=0pt
    },
    sbmarkerfade/.style={
        sbmarker,
        draw opacity=0.3,
        fill opacity=0.3
    },
    sbmarkerhollow/.style={
        diamond,
        fill=white,
        draw=red!40!black,
        line width=0.9pt,
        minimum size=5pt,
        inner sep=0pt
    },
    bracket/.style={
        very thick,
        draw=nberblue,
        opacity=0.5
    }]
    \footnotesize

  % Axis line only, with no ticks or numerical labels
  \draw[thick] (0,0) -- (1,0);

  % Bracket anchor coordinates
  \coordinate (L) at (0.375,0);
  \coordinate (R) at (0.574,0);

  % Red diamond at 0, replacing the black square
  \node[sbmarkerfade] at (0,0) {};
  \node at (0,-0.5) {$0$};

  % First-best markers: fade those outside the brackets
  \node[fbmarkerfade] at (0.22222,0) {};
  \node[fbmarker]     at (0.43,0) {};
     \node[fbmarker]       at (0.5466,0) {};
  \node[fbmarkerfade] at (0.66667,0) {};
  \node[fbmarkerfade] at (0.88889,0) {};

  % Second-best markers: fade those outside the brackets
  \node[sbmarkerfade]   at (0.15733,0) {};
  \node[sbmarker]       at (0.31466,0) {};
  \node at (0.31466,-0.5) {$L$};
  \node[sbmarkerhollow] at (0.47200,0) {};
    \node at (0.472,-0.5) {$i_k$};
  \node[sbmarker]       at (0.62933,0) {};
    \node at (0.62933,-0.5) {$R$};
  \node[sbmarkerfade]   at (0.80000,0) {};

  % Left bracket near the diamond around 0.3
  \draw[bracket]
    ([xshift=-5pt,yshift=8pt]L)
      -- ([xshift=-10pt,yshift=8pt]L)
      -- ([xshift=-10pt,yshift=-8pt]L)
      -- ([xshift=-5pt,yshift=-8pt]L);

  % Right bracket near the diamond around 0.6
  \draw[bracket]
    ([xshift=5pt,yshift=8pt]R)
      -- ([xshift=10pt,yshift=8pt]R)
      -- ([xshift=10pt,yshift=-8pt]R)
      -- ([xshift=5pt,yshift=-8pt]R);

\end{tikzpicture}
\end{center}
\caption{The agents whose signals are useful in $i_k$'s contract after Step 1.2}
\bigskip

{\scriptsize Notes: In a generic task allocation, we consider a lower bound of payment to agent $i_k$ (red hollowed diamond). The red filled diamonds represent risk-free agents, while the blue triangles represent risky agents. Step 1.2 shows that it is without loss to search among contracts that only depend on the signals of the agents between $L$ and $R$.
}
\label{fig:simple-contract}
\end{figure}

}

\medskip \noindent\textbf{Step 1.3: Lower bound of payment for agent $i_2$.} 
The relevant constraint for a lower bound for agent $i_2$ is that agent $i_2$ strictly prefers to work (i) regardless of the strategies of the noisemaker and the agents other than $i_1$ and $i_2$ but (ii) provided agent $i_1$ works.
{ As is stated in \textbf{Step 1.2}, we consider weak incentive constraints such that the inner risky agents are shirking, the outer risky agents are working, and the risk-free agent $i_1$ is working.}
This modification relaxes constraints and reduces the lower bound for $i_2$. 
These weaker constraints are written as:
{
\begin{multline*}
 \sum_\signals w_{i_2}(\signals) \Pr(\signals|\action_{i_2} = 1,\action_{i_1} = 1, \actions_{O} = \mathbf{1}, \actions_{N}= \boldsymbol{0},\shirks) -c \\
 \geq
 \sum_\signals w_{i_2}(\signals)\Pr(\signals|\action_{i_2} = 0,\action_{i_1} = 1, \actions_{O} = \mathbf{1}, \actions_{N}= \boldsymbol{0},\shirks)   
\end{multline*}
}
   for every $\shirk_{i_2}  \in \{A, B\}$ and every $\shirks_{-i_2} \in \{A,B\}^{n-1}$, { where $O$ and $N$ are the sets of outer risky agents and inner risky agents, respectively, defined by \textbf{Step 1.2}.}

   The same argument as that in {\bf Step 1.1} ensures that all of these constraints bind at the payment-minimizing contract.
   Let $ICA_{\shirks_{-i_2}}$ denote this equality constraint with $\shirk_{i_2} = A$ and $\shirks_{-i_2}$.
   Denote $ICB_{\shirks_{-i_2}}$ analogously.
Depending on the task allocation, agent $i_2$ has one or two risk-free neighbors. 
Correspondingly, we consider two cases: (a) $0\leq t_{i_1} < t_{i_2}$ (one risk-free neighbor, ``chain'') and (b) $0 \leq t_{i_2} \le  t_{i_1}$ (two risk-free neighbors, ``sandwich''). 

The chain case (a) reduces to {\bf Step 1.1}, the derivation of $i_1$'s lower bound.
Indeed, the principal is a risk-free non-neighbor who is (weakly) further away from agent $i_2$ relative to agent $i_1$, {and any risky agent on the left of $i_1$ belongs to the set of \textit{outer} risky agents}.
By {\bf Step 1.2}, it is without loss to use contracts whose payment for agent $i_2$ does not depend on the principal's signal $\signal_{0}$, {or any outer agent's signal $\signals_O$.} 
The problem is then {similar} %identical
with {\bf Step 1.1}, where agent $i_1$'s signal in the current problem plays the role of the principal's signal in {\bf Step 1.1}.
As a result, the contract that attains the lower bound pays agent $i_2$ if and only if her signal $\signal_{i_2}=\signal_{i_1}$.
The lower bound for agent $i_2$ is $\hat{K}_{i_2}(\tau, w^0) = c(1+e^{\lambda (t_{i_2} - t_{i_1})})$. {\autoref{apx:omitted_details} presents detailed calculations.}

We now consider (b) sandwich case ($0 \leq t_{i_2} \leq t_{i_1}$).
Agent ${i_2}$ has risk-free neighbors on both sides: The principal on the left and agent $i_1$ on the right.
\hyperref[apx:omitted_details]{Appendix \ref{apx:omitted_details}} shows that, 
for any fixed $\shirks_{-i_2} \in \{A, B\}^{n-1}$, the two binding constraints, ICA$_{\shirks_{-i_2}}$ and ICB$_{\shirks_{-i_2}}$, can be written as
\begin{align}
    &\frac{1}{8}P_L^-P_R^-[w_{i_2}(ABA) - w_{i_2}(AAA)] + \frac{1}{8}P_L^-P_R^+[w_{i_2}(ABB) - w_{i_2}(AAB)]\quad \nonumber
    \\
   &\qquad \qquad  + \frac{1}{8}P_L^+P_R^-[w_{i_2}(BBA) - w_{i_2}(BAA)] + \frac{1}{8} P_L^+P_R^+[w_{i_2}(BBB) - w_{i_2}(BAB)] = c, \tag{ICA$_{\shirks_{-i_2}}$} \label{ICA1.3}\\
       &\frac{1}{8}P_L^+P_R^+[w_{i_2}(AAA) - w_{i_2}(ABA)] + \frac{1}{8}P_L^+P_R^-[w_{i_2}(AAB) - w_{i_2}(ABB)] \nonumber\\
    & \qquad \qquad + \frac{1}{8}P_L^-P_R^+[w_{i_2}(BAA) - w_{i_2}(BBA)] + \frac{1}{8} P_L^-P_R^-[w_{i_2}(BAB) - w_{i_2}(BBB)] = c, \tag{ICB$_{\shirks_{-i_2}}$} \label{ICB1.3}
\end{align}
where we use the shorthand notation, $w_{i_2}(\signal_0\signal_{i_2}\signal_{i_1}):= w_{i_2}(\signal_0\signal_{i_2}\signal_{i_1};\shirks_{-i_2} )$ as well as $ P_L^+ = 1 + e^{-\lambda t_{i_2}},  P_L^- = 1 - e^{-\lambda t_{i_2}},  P_R^+ = 1 + e^{-\lambda (t_{i_1}-t_{i_2})}$, and $P_R^- = 1 - e^{-\lambda (t_{i_1}-t_{i_2})}$. The above system of equations can be reduced to one equation by substituting $w_{i_2}(BAB) - w_{i_2}(BBB)$ in ICA$_{\shirks_{-i_2}}$ using ICB$_{\shirks_{-i_2}}$, which yields  \begin{multline*}
    \frac{1}{8}[(P_L^+P_R^+)^2 - (P_L^-P_R^-)^2][w_{i_2}(AAA) - w_{i_2}(ABA)] 
    + \frac{1}{8}[P_R^+P_R^- ((P_L^+)^2 - (P_L^-)^2)][w_{i_2}(AAB) - w_{i_2}(ABB)] \\
    + \frac{1}{8}[P_L^+P_L^- ((P_R^+)^2 - (P_R^-)^2)][w_{i_2}(BAA) - w_{i_2}(BBA)] = c[P_L^-P_R^- + P_L^+P_R^+]
\end{multline*}
for any fixed $\shirks_{-i_2}$. Recall that we are minimizing $\hat{K}_{i_2}(\tau, w) = \sum_{\signals\in \{A,B\}^{n+1}}w_{i_2}(\signals)\Pr(\signals|\actions = \mathbf{1})$ subject to the above constraint. 
This is a standard linear programming, whose optimal solution is pinned down by comparing the ratios between the coefficients of the variables in the objective (i.e., $\{\Pr(\signals|\mathbf{a} = \mathbf{1})\}_{ \signals\in \{A,B\}^{n+1} }$) and the coefficients in the constraint. 

First, it is optimal to set $w_{i_2}^*(ABA) = w_{i_2}^*(ABB) = w_{i_2}^*(BBA) = 0$ because reducing these wages allows us to reduce wages for other signal profiles.  

Second, among the remaining three (i.e., $w_{i_2}(AAA)$, $w_{i_2}(AAB)$, and $w_{i_2}(BAA)$), 
$w_{i_2}(AAA)$ is the only signal profile that takes a positive value, because the relevant ratio (i.e., its coefficient in the objective to be minimized, divided by its coefficient in the constraint) is minimized:
\[
\frac{\Pr(AAA|\mathbf{a} = \mathbf{1})}{\frac{1}{8}[(P_L^+P_R^+)^2 - (P_L^-P_R^-)^2]} \leq \min\left\{\frac{\Pr(AAB|\mathbf{a} = \mathbf{1})}{\frac{1}{8}[P_R^+P_R^- ((P_L^+)^2 - (P_L^-)^2)]}, \frac{\Pr(BAA|\mathbf{a} = \mathbf{1})}{\frac{1}{8}[P_L^+P_L^- ((P_R^+)^2 - (P_R^-)^2)]}\right\}
\]
\hyperref[apx:omitted_details]{Appendix \ref{apx:omitted_details}} proves this inequality based on direct calculation.

Hence, we conclude that $w_{i_2}^*(AAB) =  w_{i_2}^*(BAA) = 0$, and $w_{i_2}^*(AAA) = \frac{4c}{e^{-\lambda t_{i_2}} + e^{-\lambda (t_{i_1}  - t_{i_2}) }}$. 
Combining the above solution with ICB$_{\shirks_{-i_2}}$, we get $w_{i_2}^*(BBB) = w_{i_2}^*(AAA)$, $w_{i_2}^*(BAB) = 0$. 
Also, the same argument applies to any $\shirks_{-i_2}$.

\medskip \noindent\textbf{Step 1.4: Lower bound of payment for agent $i_k$ with $k \ge 3$.} 
Consider any agent $i_k$ with $k  \ge 3$.
By {\bf Step 1.2}, we can focus on contracts such that the payment to agent $i_k$ does not depend on signals of her risk-free non-neighbors {or outer risky agents}.
Suppose that agent $i_k$ has only one risk-free neighbor, agent $L$, on her left side.
In this case, the problem of finding a lower bound of payment for agent $i_k$ is exactly the same as the problem for agent $i_2$ with $0 \le t_{i_1} < t_{i_2}$, where agent $t_{i_1}$ plays the role of agent $L$.
Thus,  the lower bound for agent $i_k$ is
\begin{equation}\label{equationChain}
\hat{K}_{i_k}^{\permutation,c}(\tau, w^*) = c(1+e^{\lambda (t_{i_k} - t_{L})}).
\end{equation}

If $i_k$ has two risk-free neighbors, agent $L$ on her left and agent $R$ on her right, then the problem of finding a lower bound for agent $i_k$ is the same as the problem for agent $i_2$ with $0 \le t_{i_2} \le t_{i_1}$, where the principal plays a role of agent $L$ and agent $i_1$ plays a role of agent $R$.
Thus, agent $i_k$'s lower bound is
\begin{align*}
    \hat{K}_{i_k}^{\permutation,s}(\tau, w^*) 
    = \frac{4c}{e^{-\lambda (t_{i_k} - t_{L})} + e^{-\lambda (t_{R}  - t_{i_k})}}\cdot \Pr(x_{L} = x_{i_k} = x_{R})=  c\left(1+ \frac{1+e^{-\lambda(t_{R} - t_{L})}}{e^{-\lambda (t_{i_k} - t_{L})} + e^{-\lambda (t_{R}  - t_{i_k})}}\right).
\end{align*}
Therefore, for any given elimination order, the lower bound of the total expected payment is $\hat{K}^\permutation(\tau, w^*) = \sum_{k=1}^n \hat{K}_{i_k}^\permutation(\tau, w^*)$. Denote the corresponding contract for agent $i_k$ as $w_{i_k}^*$.

\medskip \noindent\textbf{Step 2: A family of contracts $\{w^{*\epsilon}\}$ attains the lower bound as $\epsilon \rightarrow 0_+$.} Define a family of contracts $\{w^{*\epsilon}\}$ indexed by $\epsilon > 0$ as follows. For each agent $i_k$ with a single risk-free neighbor $L$ on her left side, where $L=0$ for agent $i_1$, let
\[
E_{i_k}\triangleq \{\signal_{i_k}=\signal_{L}\}.
\]
For each agent $i_k$ with two risk-free neighbors $L$ and $R$, let
\[
E_{i_k}\triangleq \{\signal_{L}=\signal_{i_k}=\signal_{R}\}.
\]
For every agent $i_k$, define
\[
w^{*\epsilon}_{i_k}(\signals)
\triangleq
w^{*}_{i_k}(\signals)+\epsilon\mathbbm{1}_{E_{i_k}}(\signals).
\]
Thus, we have $\lim_{\epsilon\rightarrow 0_+}\hat{K}^\permutation(\tau, w^{*\epsilon}) = \sum_{k=1}^n \lim_{\epsilon\rightarrow 0_+}\hat{K}_{i_k}^\permutation(\tau, w^{*\epsilon}) = \hat{K}^\permutation(\tau, w^*)$, i.e., the lower bound in \textbf{Step 1}.

\medskip \noindent\textbf{Step 3: For any $\epsilon > 0$, $w^{*\epsilon}$ RIE and induces a monitoring hierarchy.}
To show $w^{*\epsilon}$ RIE, we show that with $\epsilon > 0$, agent $i_k$ is the $k$-th player to eliminate shirking in the IESDS. First, agent $i_1$ is the first player to eliminate shirking. Indeed, facing contract $w^*$, agent $i_1$ weakly prefers to work regardless of the strategies of other agents and the noisemaker as shown in {\bf Step 1.1}. In particular, even if any of other agents are \textit{working}, their signals do not impact agent $i_1$'s incentive to work. 
Now, facing $w^{*\epsilon}$ with $\epsilon>0$, the increment of the expected wage due to working is even higher, because the probability of the event $E_{i_1}=\{\signal_0=\signal_{i_1}\}$ is strictly higher when agent $i_1$ works (note that regardless of the noisemaker's choice, the probability of $\signal_0 = \signal_{i_1}$ is $0.5$ once agent $i_1$ shirks), and the extra bonus $\epsilon$ is paid only on this event. 
Thus, agent $i_1$ strictly prefers to work regardless of the strategies of other agents and the noisemaker.

We can apply a similar argument to show that agent $i_k$ is the $k$-th player to eliminate shirking.
 For each $k=2,..., n$, we replace the principal and agent $i_k$ in the previous paragraph with the relevant risk-free neighbor(s) and agent $i_k$. If agent $i_k$ has risk-free neighbors on the left side only, the event $E_{i_k}$ is the match between agent $i_k$ and her left risk-free neighbor. Otherwise, when agent $i_k$ has risk-free neighbors on both left and right sides, $E_{i_k}$ is the triple-match event used in \textbf{Step 1.3}. In either case, \textbf{Step 1} shows that agent $i_k$ weakly prefers to work regardless of the strategies of risky agents and the noisemaker. Since $E_{i_k}$ is strictly more likely when agent $i_k$ works than when she shirks, the additional event-contingent bonus $\epsilon$ makes the preference strict. Thus, with $\epsilon > 0$, agent $i_k$ strictly prefers to work regardless of strategies of other agents and the noisemaker, given that agent $i_1, ..., i_{k-1}$'s shirking is eliminated.

Finally, we show that $w^{*\epsilon}$ induces a monitoring hierarchy consistent with \permutation. Define an ordered partition of $\mathcal{N}$ as $H^\permutation \triangleq \{ \{i_1\},  \{i_2\}, ...,  \{i_n\}\}$. For each agent $i_k$, let $M_{i_k}$ denote the set of risk-free neighbors selected in \textbf{Step 1.2}: $M_{i_k}=\{L\}$ in the one-neighbor case and $M_{i_k}=\{L,R\}$ in the two-neighbor case, where $0$ denotes the principal. By construction, $M_{i_k}\subseteq \{0,i_1,\ldots,i_{k-1}\}$, and $w^{*\epsilon}_{i_k}$ depends on other agents' signals only through the signals in $M_{i_k}$. Thus, every monitoring link into agent $i_k$ comes either from the principal or from an agent whose shirking strategy is eliminated earlier than $i_k$'s, so $w^{*\epsilon}$ induces a monitoring hierarchy $H^\permutation$.

{ It remains to connect the approximation argument to the original RIE contract.
Let $\bar w\in \contracts(\tau,\permutation)$ be the RIE contract that admits the elimination order $\permutation$, and let $\bar K=\hat K(\bar w,\tau)$.
For each strict dominance inequality used along this order, let $\Delta(\bar w)$ denote the gross expected-payment gain from working in that inequality.
Then $\Delta(\bar w)>c$ for each such inequality.
Because there are finitely many such inequalities, we can choose $\alpha\in(0,1)$ close enough to one so that}
\[
\alpha\Delta(\bar w)>c
\]
{for every inequality in the collection.
Thus $\alpha\bar w$ remains an RIE contract admitting the same order, and}
\[
\hat K(\alpha\bar w,\tau)=\alpha\hat K(\bar w,\tau)<\bar K.
\]
{The lower-bound problem in \textbf{Step 1} relaxes the constraints satisfied by every contract in $\contracts(\tau,\permutation)$, so}
\[
\hat K^\permutation(\tau,w^*)\le \hat K(\alpha\bar w,\tau)<\bar K.
\]
{ Because $\hat K(w^{*\epsilon},\tau)\to \hat K^\permutation(\tau,w^*)$ as $\epsilon\to0_+$, there exists $\epsilon>0$ such that}
\[
\hat K(w^{*\epsilon},\tau)\le \bar K.
\]
{ By \textbf{Step 3}, $w^{*\epsilon}$ is RIE and induces a monitoring hierarchy.
Therefore the original RIE contract can be replaced by a hierarchical RIE contract with weakly lower expected payment.}

\medskip \noindent\textbf{Summary.} We have shown that for any fixed task allocation $\tau$ and elimination order $\permutation$, there is a sequence of RIE contracts $\{w^{*\epsilon}\}$ that approximates the lower bound of total payments across all RIE contracts in $\mathcal{W}(\tau, \permutation)$. Moreover, each $w^{*\epsilon}$ corresponds to a monitoring hierarchy $H^\permutation$.  \hfill \qedsymbol

\subsection{Proof of \autoref{prop_inv} and \autoref{lem_chain_contract}}\label{apx:proof_inv} 

In \textbf{Step 1} of \hyperref[apx:proof-RIE]{Appendix \ref{apx:proof-RIE}}, we have derived a lower bound of the total expected payment for any fixed elimination order $\permutation=(i_1, .., i_n)$. 
In this proof, we show that for any task allocation $\tau$, the identity permutation of the elimination order, i.e., $\permutation^c = (1, 2, ..., n)$, minimizes the lower bound. Note that $\permutation^c$ corresponds to the RIE contract $w^\epsilon$ defined in \eqref{wage}.

First, suppose there are two agents, $n=2$.
 There are only two permutations: 
 Agent $1$ is $1^{st}$-ranked (elimination order $(1,2)$, ``chain'') or $2^{nd}$-ranked (elimination order $(2,1)$, ``sandwich''). The lower bound of expected payment under the two permutations are respectively: \[
    \hat{K}^{c}(t_1, t_2, w^*) = c[1+e^{\lambda t_1}] + c[1+e^{\lambda (t_2-t_1)}], \; \hat{K}^{s}(t_1, t_2, w^*)=c[1+e^{\lambda t_2}]+c\Big[1+\frac{1+e^{-\lambda t_2}}{e^{-\lambda t_1}+e^{-\lambda (t_2-t_1)}}\Big].
    \]
    For a fixed $t_2$, the two expected payments are equal when $t_1 = 0$ and $t_1 = t_2$.
     Also, $\hat{K}^s$ increases as $t_1$ approaches $t_2/2$ from left or right, while the same change decreases $\hat{K}^c$.
     Thus, for any $t_1, t_2 \in [0,1]$ such that $t_1\le t_2$, we have $\hat{K}^s(t_1, t_2) \geq \hat{K}^s(0, t_2) = \hat{K}^c(0, t_2) \geq \hat{K}^c(t_1, t_2)$.
    Therefore, the chain elimination order attains a lower value of the lower bound than the sandwich elimination order, or equivalently, the identity permutation order minimizes the lower bound for $n=2$.

    Suppose now that there are $n\ge 3$ agents.
    Let $\contract^\permutation$ denote the contract that attains the lower bound of payment under elimination order $\permutation$.
    Suppose that elimination order $\permutation$ is different from the identity permutation order $\permutation^c$.
    Let $k^*$ be the smallest step at which $\permutation$ deviates from $\permutation^c$, i.e., for any $j \le k^*-1$, $i_j = j$, but $i_{k^*} >k^*$.
    Then, we divide agents into three sets: $A = \{1,..., k^*-1\}$; $B = \{k^*,..., i_{k^*}\}$; and $C = \{i_{k^*}+1,..., n\}$.
    First, we modify the elimination order $\permutation$ so that agents in $A$ are the first $k^*-1$ agents to eliminate shirking; agents in $C$ are the last $n-i_{k^*}$ agents to eliminate shirking; and agents in $B = \{k^*,..., i_{k^*}\}$ eliminate shirking after $A$ but before $C$; also, within $B$, the elimination order is consistent with $\permutation$.
    This modification does not change the lower bound of payment, because the set of risk-free neighbors for each agent does not change at each step of elimination.

    We show that the lower bound decreases if we modify payments and elimination order for some of the agents in set $B$.
    The modification will not affect agents in $A$ or $C$.
   Thus,  to simplify notation, we assume that sets $A$ and $C$ are empty, i.e., $k^* = 1$ and $i_{k^*}=n$.
Define $k = i_2< n$.
    Then, the principal ($t=0$), agent $k$, and agent $n$ form the sandwich case of $n=2$.
    The principal can then reduce the lower bound of total payment by changing the order of elimination so that  $i_1 = k$
    and $i_2 = n$, with the rest of the elimination order following the original order.
    Thus, if a contract attains a lower bound for an elimination order that is different from the identify permutation, then we can modify the contract to strictly decrease the total payment.
    Therefore, the lower bound of the total payment is minimized by the elimination order $\permutation^c = (1, 2, ..., n)$.  
    The corresponding payment for each agent is given by \eqref{equationChain}, so the lower bound of the total expected payment is $\hat{K}^{\permutation^c}(\tau, w^*) = c\sum_{i=1}^n(1+e^{\lambda(t_i - t_{i-1})})$.

Now we show that  $w^{\epsilon_m}$ in \eqref{wage} attains the lower bound for any sequence $\{\epsilon_m\}$ such that $\epsilon_m>0$ for every $m$ and $\epsilon_m \rightarrow 0$. 
By the definition of $w^\epsilon$, we have
\[
\hat{K}(w^{\epsilon_m}, \tau) = \sum_{i=1}^n \sum_{\signals \in \{A,B\}^{n+1}} w_i^{\epsilon_m}(\signals) \Pr(\signals|\tau, \actions = \mathbf{1}) =  \sum_{i=1}^n c(1+e^{\lambda(t_i - t_{i-1})}) + \epsilon_m\sum_{i=1}^n\Pr(x_{i} = x_{i-1})
\]
Thus, we have $\lim_{\epsilon_m\rightarrow 0}\hat{K}(w^{\epsilon_m}, \tau) = \sum_{i=1}^n c(1+e^{\lambda(t_i - t_{i-1})})$, i.e., the lower bound. This proves \autoref{prop_inv}. Finally, $w^\epsilon$ is an RIE contract for any $\epsilon>0$, which follows exactly the same argument as \textbf{Step 3} of \hyperref[apx:proof-RIE]{Appendix \ref{apx:proof-RIE}}. This proves \autoref{lem_chain_contract}. 
   \hfill \qedsymbol

\subsection[Proof of Claim 1 and worst-case uniform noise]{Proof of \hyperref[claim_bni]{Claim \ref{claim_bni}} and worst-case uniform noise}
\label{appendix:bni-worst-noise}
\label{appendixClaim1}

Fix a task profile $\tau$ with $0=t_0\le t_1\le\cdots\le t_n$ and maintain the notation of Section~\ref{sec:BNI}. Because the state process and the signal space are invariant to exchanging the labels $A$ and $B$, we may assume without loss that the noisemaker assigns $A$ to a shirking agent $i$ with probability $q_i\in[0,1/2]$. It is convenient to set $t_{n+1}:=\infty$, so $r_n^+:=e^{-\lambda(t_{n+1}-t_n)}=0$.

\begin{claim}[Optimal BNI contracts and worst-case noisemaker]\label{claim:bni-worst-noise}
For each $i\le n$, define
\[
r_i^-:=e^{-\lambda(t_i-t_{i-1})},
\qquad
r_i^+:=e^{-\lambda(t_{i+1}-t_i)}.
\]
Let
\[
f_i^*=
\frac{(1+r_i^-)(1+r_i^+)}{2(1+r_i^-r_i^+)},
\qquad
\chi_i:=
\begin{cases}
1, & i<n,\\
2, & i=n.
\end{cases}
\]
The minimized one-agent BNI expected payment is
\[
W_i(q_i)
=
\frac{c f_i^*}{f_i^*-q_i}.
\]
Define
\[
\mathcal A_i:=
\begin{cases}
\{\signal_{i-1}=\signal_i=\signal_{i+1}=A\}, & i<n,\\
\{\signal_{n-1}=\signal_n=A\}, & i=n,
\end{cases}
\]
and
\[
\mathcal M_i:=
\begin{cases}
\{\signal_{i-1}=\signal_i=\signal_{i+1}\}, & i<n,\\
\{\signal_{n-1}=\signal_n\}, & i=n.
\end{cases}
\]
If $q_i<1/2$, one optimal contract pays only on $\mathcal A_i$ with bonus
\[
\frac{4c}{\chi_i(1+r_i^-r_i^+)(f_i^*-q_i)}.
\]
If $q_i=1/2$, the symmetric match contract is optimal: pay agent $i$ if and only if $\mathcal M_i$ occurs, with bonus
\[
b_i=\frac{4c}{\chi_i(r_i^-+r_i^+)}
=
\begin{cases}
\dfrac{4c}{e^{-\lambda(t_i-t_{i-1})}+e^{-\lambda(t_{i+1}-t_i)}}, & i<n,\\[8pt]
2c\,e^{\lambda(t_n-t_{n-1})}, & i=n.
\end{cases}
\]
At $q_i=1/2$, the minimized expected payment is
\[
W_i^*
=
\frac{c(1+r_i^-)(1+r_i^+)}{r_i^-+r_i^+}.
\]

For every agent, $W_i(q_i)$ is strictly increasing on $[0,1/2]$ and hence is maximized at $q_i=1/2$. Thus, uniform draw from $\{A, B\}$ is worst for the principal under BNI.
\end{claim}

\begin{proof}
All agents working is a Bayesian Nash equilibrium if and only if each one-deviation incentive constraint holds:
\begin{equation}\label{eq:bni-one-dev}
\E[\contract_i(\signals)\mid\actions=\mathbf{1}]
-c
\ge
\E[\contract_i(\signals)\mid\action_i=0,\,\actions_{-i}=\mathbf{1}].
\end{equation}
Because $\contract_i$ enters only agent $i$'s constraint and agent $i$'s expected payment, minimizing the all-work expected payment decomposes across agents. Define agent $i$'s local statistic by
\[
Z_i:=
\begin{cases}
(\signal_{i-1},\signal_i,\signal_{i+1})\in\{A,B\}^3, & i<n,\\
(\signal_{n-1},\signal_n)\in\{A,B\}^2, & i=n.
\end{cases}
\]
Conditional on the truthful neighboring signals included in $Z_i$, farther signals are independent of agent $i$'s action. Averaging any contract over farther signals therefore preserves the incentive constraint and weakly lowers the expected payment, so no information outside $Z_i$ is needed.

\medskip
\noindent\textbf{The one-agent linear program.}
Let $z$ denote a realization of $Z_i$, and let all sums over $z$ range over $\{A,B\}^3$ if $i<n$ and over $\{A,B\}^2$ if $i=n$. Write $P_i^1(z)$ for the probability of $Z_i=z$ when all agents work, and $P_i^0(z;q_i)$ for the probability of $Z_i=z$ when agent $i$ shirks, all other agents work, and $\Pr(\shirk_i=A)=q_i\le 1/2$. Agent $i$'s problem is
\[
\min_{\contract_i\ge 0}
\sum_z \contract_i(z)P_i^1(z)
\quad\text{s.t.}\quad
\sum_z \contract_i(z)\bigl(P_i^1(z)-P_i^0(z;q_i)\bigr)\ge c.
\]
Cells with $P_i^1(z)=0$ receive no positive payment in an optimum. For cells with $P_i^1(z)>0$, define
\[
\rho_i(z;q_i):=\frac{P_i^0(z;q_i)}{P_i^1(z)}.
\]
The problem has a single incentive constraint and limited liability. Thus, for any feasible contract,
\[
\sum_z \contract_i(z)P_i^1(z)
\ge
\frac{c}{1-\rho_i^*(q_i)},
\qquad
\rho_i^*(q_i):=\min_{z:P_i^1(z)>0}\rho_i(z;q_i),
\]
and equality is attained by concentrating all payment on cells that minimize $\rho_i(z;q_i)$. Therefore the value is
\[
W_i(q_i)
=
\frac{c}{1-\rho_i^*(q_i)}.
\]

\medskip
\noindent\textbf{Likelihood ratios.}
First take $i<n$. For adjacent truthful signals $a,b\in\{A,B\}$ with positive probability, define
\[
f_i(a,b)
:=
\Pr(\signal_i=A\mid \signal_{i-1}=a,\,\signal_{i+1}=b,\,\actions=\mathbf{1}).
\]
The transition correlations $r_i^-,r_i^+$ lie in $(0,1]$ for $i<n$, so the formulas below allow collocated task locations. By Bayes' rule and the Markov property,
\[
f_i(A,A)
=
\frac{(1+r_i^-)(1+r_i^+)}{2(1+r_i^-r_i^+)}
=f_i^*.
\]
The same calculation gives $1-f_i(B,B)=f_i^*$. For $(\signal_{i-1},\signal_{i+1})=(A,B)$ or $(B,A)$, whenever the conditioning event has positive probability,
\[
f_i(A,B)=\frac{(1+r_i^-)(1-r_i^+)}{2(1-r_i^-r_i^+)},
\qquad
f_i(B,A)=\frac{(1-r_i^-)(1+r_i^+)}{2(1-r_i^-r_i^+)}.
\]
If either of these conditioning events has zero probability, it is omitted from the likelihood-ratio minimization. Otherwise,
\[
f_i^*-f_i(A,B)
=
\frac{r_i^+(1-(r_i^-)^2)}
{(1+r_i^-r_i^+)(1-r_i^-r_i^+)}
\ge 0,
\]
and
\[
f_i^*-f_i(B,A)
=
\frac{r_i^-(1-(r_i^+)^2)}
{(1+r_i^-r_i^+)(1-r_i^-r_i^+)}
\ge 0.
\]
Thus the maximum of $f_i(a,b)$ is attained at $a=b=A$, and by symmetry the maximum of $1-f_i(a,b)$ is attained at $a=b=B$. The common maximum is $f_i^*$. For agent $n$, the same expression with $r_n^+=0$ gives
\[
f_n^*
=
\Pr(\signal_n=A\mid\signal_{n-1}=A,\,\actions=\mathbf{1})
=
\frac{1+r_n^-}{2}.
\]

For a statistic $z=(a,s,b)$ with positive truthful probability,
\[
\rho_i(a,s,b;q_i)
=
\begin{cases}
q_i/f_i(a,b), & s=A,\\[4pt]
(1-q_i)/(1-f_i(a,b)), & s=B.
\end{cases}
\]
Because $q_i\le 1/2$, $f_i(a,b)\le f_i^*$, and $1-f_i(a,b)\le f_i^*$, the smallest likelihood ratio is attained on an $A$-cell with the largest posterior probability of $\signal_i=A$.
The same conclusion holds in the two-cell case for agent $n$. Hence, for every $i\le n$,
\[
\rho_i^*(q_i)
=
\frac{q_i}{f_i^*},
\]
and therefore
\[
W_i(q_i)
=
\frac{c f_i^*}{f_i^*-q_i}.
\]
$W_i(q_i)$ is strictly increasing on $[0,1/2]$ and is maximized at $q_i=1/2$.

For $q_i<1/2$, the contract that pays agent $i$ only on $\mathcal A_i$ is optimal because $\mathcal A_i$ minimizes $\rho_i(\cdot;q_i)$; under collocation, other cells may also minimize the same likelihood ratio. The bonus formula is obtained by making the single incentive constraint bind:
\[
P_i^1(\mathcal A_i)-P_i^0(\mathcal A_i;q_i)
=
\chi_i\frac{1+r_i^-r_i^+}{4}(f_i^*-q_i).
\]

Now specialize to uniform noise. For every $i\le n$,
\[
\Pr(\mathcal M_i\mid\actions=\mathbf{1})
=
\chi_i\frac{(1+r_i^-)(1+r_i^+)}{4},
\]
whereas
\[
\Pr(\mathcal M_i\mid\action_i=0,\,\actions_{-i}=\mathbf{1})
=
\chi_i\frac{1+r_i^-r_i^+}{4}.
\]
Thus the incentive gain from a unit bonus on $\mathcal M_i$ is
\[
\chi_i\frac{r_i^-+r_i^+}{4}.
\]
The bonus
\[
b_i
=
\frac{4c}{\chi_i(r_i^-+r_i^+)}
\]
satisfies the incentive constraint with equality. Its expected payment is
\[
b_i\Pr(\mathcal M_i\mid\actions=\mathbf{1})
=
\frac{c(1+r_i^-)(1+r_i^+)}{r_i^-+r_i^+}
=
W_i(1/2),
\]
so the symmetric match contract is optimal.

\medskip
\noindent\textbf{Combining agents.}
Applying the corresponding one-agent optimum to each $\contract_i$ gives an optimal BNI contract. Substituting $\chi_i=1$ for $i<n$ and $(\chi_n,r_n^+)=(2,0)$ for agent $n$ gives the displayed bonus and expected-payment formulas. This proves \hyperref[claim_bni]{Claim \ref{claim_bni}}.
\end{proof}

\subsection{Proof of \autoref{lemma_policy}}\label{appendixLemma1}
Take any task allocation $\tau = (t_1,..., t_n)$, and suppose that all agents exert effort.
Given a signal profile $\signals = (x(0), x(t_1),..., x(t_n))$, the principal's problem of choosing an action policy is $
 \min_p \mathbb{E}\left[\int_0^1 \mathbbm{1}(p(t) \neq x(t)) \, \dint  t \mid \tau, \signals\right]$. 
The integrand $\mathbbm{1}(p(t) \neq x(t))$ is bounded at each $t$.
Thus, we can ignore the set of locations at which the principal's action changes---which has zero measure---and rewrite the problem as a pointwise optimization, i.e.,  $
 \min_{p(t)\in \{A, B\}} \mathbb{E}[\mathbbm{1}(p(t) \neq x(t)) |\tau, \signals], \forall t\in [0,1]$. 
Note that $ \mathbb{E}[\mathbbm{1}(p(t) \neq x(t))|\tau, \signals] = \Pr(p(t) \neq x(t)|\tau, \signals)$.

Take any location $t \in [0,t_n]$, and let $i \in \{1,...,n\}$ satisfy $t \in [t_{i-1}, t_i]$, where $t_0 =0$.
Due to the Markov property of the state distribution, probability $\Pr(p(t) \neq x(t)|\tau, \signals)$ depends only on $\state(t_{i-1})$ and $\state(t_{i})$.
If $\state(t_{i-1})\neq \state(t_i)$, then $\Pr(x(t) = x(t_{i})|\tau, \signals) \geq \Pr(x(t) = x(t_{i-1})|\tau, \signals) \iff |t-t_{i}| \leq |t-t_{i-1}|$; if $\state(t_{i-1})=\state(t_i)$, both neighboring signals recommend the same action. 
Thus, the optimal action for location $t$ is $p(t) =\state(t^*)$, where $t^* \in \argmin_{t'\in \{t_{i-1}, t_i\}} |t' -t| = \argmin_{t'\in \{0,t_1,...,t_n\}} |t' -t|$.
This formula also applies to locations between $t_n$ and $1$, which do not have their right-hand neighbor. It remains to average the minimized posterior loss over truthful signal profiles to obtain the information loss in \eqref{inf_loss}. The detail is as follows. Let $\Delta_i=t_i-t_{i-1}$. Under the nearest-signal action, a point at distance $u$ from the used sampled location is misclassified exactly when its local state differs from that sampled local state. Averaging over truthful signal profiles, this occurs with probability $\frac{1}{2}(1-e^{-\lambda u})$ by \eqref{cor}. Hence, on the interval $[t_{i-1},t_i]$, the nearest-signal policy generates expected loss
\begin{align*}
&\int_{t_{i-1}}^{(t_{i-1}+t_i)/2}\frac{1-e^{-\lambda(t-t_{i-1})}}{2}\,dt
+\int_{(t_{i-1}+t_i)/2}^{t_i}\frac{1-e^{-\lambda(t_i-t)}}{2}\,dt =\int_0^{\Delta_i/2}(1-e^{-\lambda u})\,du
=\frac{\Delta_i}{2}-\frac{1-e^{-\lambda \Delta_i/2}}{\lambda}.
\end{align*}
The remaining interval $[t_n,1]$ is evaluated using the signal at $t_n$, so its loss is
\[
\int_{t_n}^{1}\frac{1-e^{-\lambda(t-t_n)}}{2}\,dt
=\frac{1-t_n}{2}-\frac{1-e^{-\lambda(1-t_n)}}{2\lambda}.
\]
Summing these terms and using $\sum_{i=1}^n\Delta_i=t_n$ gives
\begin{align*}
L(\tau)
&=\sum_{i=1}^n\left[\frac{\Delta_i}{2}-\frac{1-e^{-\lambda \Delta_i/2}}{\lambda}\right]
+\frac{1-t_n}{2}-\frac{1-e^{-\lambda(1-t_n)}}{2\lambda}
=\frac{1}{2}-\frac{2n+1}{2\lambda}
+\frac{1}{\lambda}\left[
\sum_{i=1}^n e^{-\lambda (t_i-t_{i-1})/2}
+\frac{1}{2}e^{-\lambda(1-t_n)}
\right],
\end{align*}
which is \eqref{inf_loss}.
\hfill $\square$

\subsection{Proof of \autoref{prop_policy}}
The information loss \eqref{inf_loss} in \autoref{lemma_policy} is strictly convex and symmetric in $\Delta t_i =t_i - t_{i-1}, \forall i = 1,..., n$.
Thus, for a given $t_n$, we should choose $t_1,..., t_{n-1}$ to equalize $\Delta t_i$.
Let $\Delta t$ be the equalized value with
$n \Delta t = t_n$.
Plugging $t_n = n\Delta t$ and $t_i - t_{i-1} = \Delta t$ into \eqref{inf_loss} and minimizing it with respect to $\Delta t$, we obtain 
$\Delta t^\dag = \frac{2}{2n + 1}$, or equivalently, $
t_i^\dag = \frac{2i}{2n + 1}, \forall i = 1, 2,..., n$. Then, the first-best information loss is $
L(\tau^\dag) = \frac{1}{2} + \frac{2n+1}{2\lambda}\left(e^{-\frac{\lambda}{2n + 1}} - 1\right)$. Denote $x^\dag \triangleq \lambda/(2n+1)$ and $y^\dag \triangleq e^{-x^\dag}<1$. The envelope theorem implies $\frac{\partial L(\tau^\dag)}{\partial n} = \frac{y^\dag-1}{\lambda}+\frac{y^\dag}{2n+1}<0$ and $\frac{\partial L(\tau^\dag)}{\partial \lambda} = \frac{(2n+1)(1-y^\dag)-\lambda y^\dag}{2\lambda^2}>0$. To verify the signs, rewrite the first derivative as $\{(1+x^\dag)e^{-x^\dag}-1\}/\lambda$ and the second derivative as $(2n+1)\{1-(1+x^\dag)e^{-x^\dag}\}/(2\lambda^2)$. Since $e^{x^\dag}>1+x^\dag$ for $x^\dag>0$, we have $(1+x^\dag)e^{-x^\dag}<1$, so the first derivative is negative and the second is positive.
\hfill $\square$

\subsection{Proof of \autoref{prop_task}}\label{appendixProp3}
\autoref{lemma_policy} and  \autoref{prop_inv} imply that the task allocation problem is
\begin{equation}\label{equationTask}
\min_{\tau} \frac{1}{2} - \frac{2n+1}{2\lambda} + \frac{1}{\lambda}\left[\sum_{i=1}^n e^{-\frac{\lambda \Delta t_i}{2}} + \frac{1}{2}e^{-\lambda (1-\sum_{i=1}^n\Delta t_i)}\right] + c\sum_{i=1}^n [1+e^{\lambda \Delta t_i}],
\end{equation}
which is strictly convex and symmetric in the choice variables, $\Delta t_1,...,\Delta t_n$.
As a result, there is a unique $\Delta t \in [0,1)$ such that $\Delta t_i = \Delta t$ for every $i$.
Compared to the first best, which minimizes the information loss, the principal's objective has the incentive cost, $c\sum_{i=1}^n [1+e^{\lambda \Delta t_i}] = cn [1+e^{\lambda \Delta t}]$.
This term is strictly increasing in $\Delta t$.
The argument of monotone comparative statics implies that the optimal $\Delta t$, denoted by $\Delta t^*$, is strictly less than the first-best solution, $\Delta t^\dag$.
\hfill $\square$

\subsection{Proof of \autoref{coro_task}
}\label{appendixCor2}

We consider two cases.
If $c<  \bar{c} \triangleq \frac{1-e^{-\lambda}}{2\lambda}$, then 
the optimal $\Delta t^*$ is strictly positive and uniquely solves the first-order condition
$x^{2(n-1)} + 2\lambda e^\lambda c -\frac{e^\lambda}{x^3}=0$,
where $x\triangleq e^{\frac{\lambda \Delta t^*}{2}}\geq 1$.
The left-hand side is strictly increasing in $n$ and $c$, so to maintain the equation, the optimal $\Delta t^*$ must decrease.
If $c\ge  \bar{c}$, then regardless of $n$, we have $\Delta t^* =0$ at the optimum. 
Therefore, the optimal $\Delta t^*$ is overall (weakly) decreasing in $c$ and $n$.

As for the comparative statics of $n\Delta t^*$,   we study its derivative with respect to $n$. We are to show that $\frac{\partial (n\Delta t^*)}{\partial n} = \frac{d\Delta t^*}{dn} \cdot n + \Delta t^* > 0$, which, by implicitly differentiating $\Delta t^*$ with respect to $n$ on the optimal condition, boils down to $
    3e^{\lambda (1-\frac{3}{2}\Delta t^*)} - 2e^{\lambda (n-1)\Delta t^*} > 0$. 
    This inequality always holds because $
    e^{\lambda (1-\frac{3}{2}\Delta t^*)} - e^{\lambda (n-1)\Delta t^*} = e^{\lambda (n-1)\Delta t^*}[e^{\lambda (1-\frac{3}{2}\Delta t^* + \Delta t^* - n\Delta t^*)} - 1] = e^{\lambda [1-(\frac{1}{2}+n)\Delta t^*]} - 1 > 0$. 
    For the rationale of the last inequality, recall that \autoref{prop_policy} implies $(\frac{1}{2}+n)\Delta t^\dag =(\frac{1}{2}+n) \cdot \frac{2}{2n+1} = 1$.
    By \autoref{prop_task}, $\Delta t^* < \Delta t^\dag$. Hence, the last inequality is true.
\hfill
$\square$

\subsection{Proof of  \autoref{random_eff2}}\label{appendix:random_task}
\label{app_random}
For each $p\in (0,1)$, define task allocation policy $(T^{p},\pi^{p})$ as follows: $T^p = \{\tau^\dag, \tau^\dag_1,..., \tau^\dag_n\}$, and 
\begin{equation}\label{random1}
\pi^p(\tau)\triangleq
\begin{cases}
1-p &\mbox{ if } \tau=\tau^{\dag}\\
\frac{p}{n} & \mbox{ if } \tau=\tau^{\dag}_i\triangleq(t_1^{\dag},...,t_{n-1}^{\dag}, t_i^{\dag}), i=1,2,...,n-1\\
\frac{p}{n} & \mbox{ if } \tau=\tau^{\dag}_n\triangleq(t_1^{\dag}, ..., t_{n-2}^{\dag}, t_n^{\dag},t_n^{\dag}).
\end{cases}
\end{equation}
This task allocation policy assigns all agents to their respective first-best locations, $t_1^{\dag},..., t_n^{\dag}$, with probability $1-p$.
With probability $p/n$ each, the allocation policy chooses $\tau^{\dag}_i$, which continues to assign agents $1,..., n-1$ to their first-best locations but now assigns agent $n$ to $t_i^{\dag}$ so that her signal is used to monitor agent $i$.
With the remaining probability $p/n$, the policy draws $\tau^{\dag}_n$, which assigns all agents but $n-1$ to their first-best locations and agent $n-1$ to $t_n^{\dag}$, so that agent $n-1$'s signal can be used to monitor agent $n$.
As a result, when agent $i\in \{1,..., n\}$ learns that her own assignment is $t^{\dag}_i$, she is uncertain whether the realized task allocation is $\tau^{\dag}$ or $\tau^{\dag}_i$, i.e., whether another agent is assigned to the same location as $i$.
When agent $i\in\{n-1,n\}$ is assigned to location $t^{\dag}_j\neq t^{\dag}_i$, she knows that she is privately appointed to monitor agent $j$ and her signal will be compared to the principal's or the left-neighboring agent's (risk-free) signal.

The following result presents a contract that, combined with the task allocation policy in \eqref{random1}, robustly implements the desired outcome.
\begin{lemma}\label{prop_random2}
Suppose that there are $n\ge 2$ agents. 
Assume that each agent observes the task allocation, the contract, and the realization of her assigned location, but not the realized locations of other agents.
Take any $p \in (0,1)$ and consider a stochastic task allocation policy $(T^{p},\pi^{p})$ defined by \eqref{random1}. For any $\epsilon>0$, the following contract $w^{p,\epsilon}$ is RIE:

For each agent $i=1,..., n-2$,
\begin{equation}\label{random_wage3}
w_i^{p,\epsilon}(\tau,\signals)=
\begin{cases}
\tfrac{2nc}{p}+\epsilon & \mbox{ if } \tau=\tau^{\dag}_i, \signal_i  = \signal_n,\\
0 &\mbox{ otherwise; }
\end{cases},
\end{equation}
for agent $n-1$,
\begin{equation}\label{random_wage4}
w_{n-1}^{p,\epsilon}(\tau, \signals)=
\begin{cases}
\tfrac{2nc}{p}+\epsilon & \mbox{ if } \tau=\tau^{\dag}_{n-1}, \signal_{n-1}=\signal_n,\\
2ce^{\lambda (t^{\dag}_n-t^{\dag}_{n-2})}+\epsilon &\mbox{ if } \tau=\tau^{\dag}_n, \signal_{n-1}=\signal_{n-2}\\
0 &\mbox{ otherwise;}
\end{cases},
\end{equation}
 and for agent $n$,
\begin{equation}\label{random_wage5}
w_n^{p,\epsilon}(\tau, \signals)=
\begin{cases}
\tfrac{2nc}{p}+\epsilon &\mbox{ if } \tau=\tau^{\dag}_n, \signal_{n}=\signal_{n-1}\\
2ce^{\lambda (t^{\dag}_i-t^{\dag}_{i-1})}+\epsilon &\mbox{ if } \tau=\tau^{\dag}_i, \signal_{n}=\signal_{i-1}, i<n\\
0 &\mbox{ otherwise. }
\end{cases}.
\end{equation}
The principal's incentive cost is $K(0) + \frac{p}{n}(2ce^{2\lambda t_1^\dag} + 2(n-1)ce^{\lambda t_1^\dag}-2nc)$, as if he can perfectly verify whether each agent's signal matches her assigned location state, except for the events in which he needs to provide strict incentive to the monitor.
\end{lemma}

\begin{proof}[Proof of \autoref{prop_random2}]
For agent $i$ with $i = 1, 2, ..., n-2$, her incentive cost is determined by the following linear programming. 
\[
\inf_{w_i(\tau_i^\dag), w_i(\tau \neq \tau_i^\dag)} \frac{p}{n}w_i(\tau_i^\dag) + \left(1-\frac{p}{n}\right)\frac{1}{2}(1+e^{-\lambda (t_i - t_{i-1})}) w_i(\tau \neq \tau_i^\dag)
\]
subject to \[
\frac{p}{n}w_i(\tau_i^\dag) + \left(1-\frac{p}{n}\right)\frac{1}{2}(1+e^{-\lambda (t_i - t_{i-1})}) w_i(\tau \neq \tau_i^\dag) - c > \frac{p}{n} \cdot \frac{1}{2}w_i(\tau_i^\dag) + \left(1-\frac{p}{n}\right)\frac{1}{2}w_i(\tau \neq \tau_i^\dag) 
\]
which leads to \eqref{random_wage3}. 
In this construction, agent 1's dominance incentive is established in the \emph{second} round of IESDS,\footnote{
We consider the agent normal form in the IESDS process, i.e., different types of an agent are treated as many agents.
} after agent $n$ with $\tau_1^\dag$. Moreover, each agent $i$'s dominance incentive is established in the $2i$-th round, immediately after agent $n$ with $\tau_i^\dag$, for all $i = 2, ..., n-2$.
Agent $n-1$ has two private types, the ``monitor'' type (when $\tau = \tau_{n}^\dag$) and the ``worker'' type (when $\tau \neq \tau_{n}^\dag$). By construction, when $\tau_{n}^\dag$ is realized, agent $n$'s dominance incentive is established \emph{after} $n-1$. Therefore, the monitor-type agent $n-1$ has a simple incentive cost under chain monitoring, and gets paid if and only if $\signal_{n-1} = \signal_{n-2}$ despite being located at $t_n$. The worker-type agent $n-1$ has the same incentive cost as agent $i = 1, 2, ..., n-2$. Combining both types, agent $n-1$'s wage scheme satisfies \eqref{random_wage4}. Finally, agent $n$ has $n-1$ (private) monitor types (when $\tau \neq \tau_n^\dag$) and one worker type (when $\tau = \tau_n^\dag$). In each monitor type with $\tau = \tau_i^\dag$, agent $n$ is paid if and only if her signal matches with agent $i-1$, whose dominance incentive is established one round before. Particularly, in the monitor type with $\tau = \tau_1^\dag$, agent $n$'s signal matches with the principal, establishing her dominance incentive in the \emph{first} round of IESDS. In the worker type, agent $n$ has the same incentive cost as agent $i = 1, 2, ..., n-2$. Combining all types, agent $n$'s wage scheme satisfies \eqref{random_wage5}. Thus, the total expected incentive cost is \begin{align*}
    &\lim_{\epsilon \rightarrow 0} \sum_{i=1}^{n} \frac{p}{n}\cdot (\frac{2nc}{p} + \epsilon) - 2pc + \frac{p}{n}\cdot [c(1+e^{2\lambda t_1^\dag}) + \epsilon] + \frac{n-1}{n}p \cdot [c(1+e^{\lambda t_1^\dag}) + \epsilon] \\
    &= K(0) + \frac{p}{n}(c(1+e^{2\lambda t_1^\dag}) + (n-1)c(1+e^{\lambda t_1^\dag})-2nc)
\end{align*}
The expression involving $t_1^\dag$ follows from $t_1^\dag = t_i^\dag - t_{i-1}^\dag$, $\forall i$.
\end{proof}

We now prove \autoref{random_eff2}. For an arbitrary constant $\delta \in (0, 1)$, define $p_k = \delta^k$, $k \in \mathbb{N}_+$, and choose any sequence $\epsilon_k>0$ with $\epsilon_k\rightarrow 0$. We obtain a sequence of probabilities $\{p_k\}_{k=1}^\infty$. Define $(T^{p_k}, \pi^{p_k})$ by \eqref{random1}.  \autoref{prop_random2} implies that, for each $k$, the contract $w^{p_k,\epsilon_k}$ is RIE. The resulting total loss differs by terms that vanish with $\epsilon_k$ from
\[
(1-p_k) L(\tau^\dag) + \frac{n-1}{n}p_k [L(\tau^\dag)+const]+ \frac{1}{n}p_k [L(\tau^\dag)+const] + K(0) + \frac{p_k}{n}(2ce^{2\lambda t_1^\dag} + 2(n-1)ce^{\lambda t_1^\dag}-2nc).
\]
As $k\rightarrow \infty$, we obtain the limit in \autoref{random_eff2}, i.e., $L(\tau^\dag) + K(0)$.
\hfill $\square$

\newpage

\section[Appendix: Omitted Details in Appendix A.1]{Appendix: Omitted Details in \hyperref[apx:proof-RIE]{Appendix \ref{apx:proof-RIE}}}\label{apx:omitted_details}
In {\bf Step 1.1} of \hyperref[apx:proof-RIE]{Appendix \ref{apx:proof-RIE}}, constraints \eqref{equationICA} and \eqref{equationICB} for agent $i_1$ can be rewritten as \begin{align}
        &\sum_\signals w_{i_1}(\signals) [\Pr(\signals|\action_{i_1} = 1, \actions_{-i_1}= \boldsymbol{0},\shirk_{i_1} = A, \shirks_{-i_1}) - \Pr(\signals|\action_{i_1} = 0,\actions_{-i_1}= \boldsymbol{0},\shirk_{i_1} = A, \shirks_{-i_1})] \geq c \label{ICA-int1}\\
        &\sum_\signals w_{i_1}(\signals) [\Pr(\signals|\action_{i_1} = 1, \actions_{-i_1}= \boldsymbol{0},\shirk_{i_1} = B, \shirks_{-i_1})-\Pr(\signals|\action_{i_1} = 0,\actions_{-i_1}= \boldsymbol{0},\shirk_{i_1} = B, \shirks_{-i_1})] \geq c \label{ICB-int1}
\end{align}
  $\forall \ \shirks_{-i_1} \in \{A,B\}^{n-1}$. Let $w_{i_1}(\signal_0\signal_{i_1}):=w_{i_1}(\signal_0\signal_{i_1}; \shirks_{-i_1})$ 
  for every fixed $\shirks_{-i_1}$. Fix $\shirk_{i_1} = A$. Enumerate $[\Pr(\signals|\action_{i_1} = 1, \actions_{-i_1}= \boldsymbol{0},\shirks) - \Pr(\signals|\action_{i_1} = 0,\actions_{-i_1}= \boldsymbol{0},\shirks)]$ for each combination of $\signal_0\signal_{i_1}$ as:
\begin{align*}
    &\Pr(AA; \signals_{-i_1}|\action_{i_1} = 1, \actions_{-i_1}= \boldsymbol{0},\shirk_{i_1} = A,\shirks_{-i_1}) - \Pr(AA; \signals_{-i_1}|\action_{i_1} = 0,\actions_{-i_1}= \boldsymbol{0},\shirk_{i_1} = A,\shirks_{-i_1}) \\
    &= \underbrace{\frac{1}{2}}_{\Pr(\signal_0 = A)}\times \underbrace{\frac{1}{2}(1+e^{-\lambda t_{i_1}})}_{\Pr(\state(0) = \state(t_{i_1}))} - \underbrace{\frac{1}{2}}_{\Pr(\signal_0 = A)}\times \underbrace{1}_{\Pr(\signal_{i_1}=A|\action_{i_1} = 0,\shirk_{i_1} = A)} = -\frac{1}{4}(1 - e^{-\lambda t_{i_1}})\\
    &\Pr(AB; \signals_{-i_1}|\action_{i_1} = 1, \actions_{-i_1}= \boldsymbol{0},\shirk_{i_1} = A,\shirks_{-i_1}) - \Pr(AB; \signals_{-i_1}|\action_{i_1} = 0,\actions_{-i_1}= \boldsymbol{0},\shirk_{i_1} = A,\shirks_{-i_1})  \\
    &= \frac{1}{2} \times \frac{1}{2}(1-e^{-\lambda t_{i_1}}) - \frac{1}{2}\times \underbrace{0}_{\Pr(\signal_{i_1} = B|\action_{i_1} = 0,\shirk_{i_1} = A)} =\frac{1}{4}(1-e^{-\lambda t_{i_1}})  \\
    &\Pr(BA; \signals_{-i_1}|\action_{i_1} = 1, \actions_{-i_1}= \boldsymbol{0},\shirk_{i_1} = A,\shirks_{-i_1}) - \Pr(BA; \signals_{-i_1}|\action_{i_1} = 0,\actions_{-i_1}= \boldsymbol{0},\shirk_{i_1} = A,\shirks_{-i_1}) 
    = -\frac{1}{4}(1+e^{-\lambda t_{i_1}})\\
    &\Pr(BB; \signals_{-i_1}|\action_{i_1} = 1, \actions_{-i_1}= \boldsymbol{0},\shirk_{i_1} = A,\shirks_{-i_1}) - \Pr(BB; \signals_{-i_1}|\action_{i_1} = 0,\actions_{-i_1}= \boldsymbol{0},\shirk_{i_1} = A,\shirks_{-i_1}) 
    =  \frac{1}{4}(1+e^{-\lambda t_{i_1}})
\end{align*}
Substituting in \eqref{ICA-int1} and rearranging, we get \eqref{ICA1.1}.
 The calculations for \eqref{ICB1.1} is symmetric.

\medskip

{
To see that the chain case of \textbf{Step 1.3} ($0 \leq t_{i_1} \leq t_{i_2}$) can reduce to \textbf{Step 1.1}, note that agent $i_2$'s contract must satisfy the following set of constraints \begin{multline*}
 \sum_\signals w_{i_2}(\signals) [\Pr(\signals|\action_{i_2} = 1,\action_{i_1} = 1, \actions_{O} = \mathbf{1}, \actions_{N}= \boldsymbol{0},\shirks) - \Pr(\signals|\action_{i_2} = 0,\action_{i_1} = 1, \actions_{O} = \mathbf{1}, \actions_{N}= \boldsymbol{0},\shirks)  ]
 \geq
 c 
\end{multline*}
for every $\shirk_{i_2}  \in \{A, B\}$ and every $\shirks_{-i_2} \in \{A,B\}^{n-1}$. Since the outer risky agents and agent $i_1$ are working, the only relevant constraints are indexed by  $\shirk_{i_2}  \in \{A, B\}$ and $\shirks_{N} \in \{A,B\}^{|N|}$. Denote $P^*(\signals) := \Pr(\signals|\actions = \mathbf{1})$.
Fix each combination of  $(\shirk_{i_2}, \shirks_{N})$. The likelihood difference is \begin{align*}
    &[\Pr(\signals|\action_{i_2} = 1,\action_{i_1} = 1, \actions_{O} = \mathbf{1}, \actions_{N}= \boldsymbol{0},\shirks) - \Pr(\signals|\action_{i_2} = 0,\action_{i_1} = 1, \actions_{O} = \mathbf{1}, \actions_{N}= \boldsymbol{0},\shirks)  ] \\
    &=  P^*(\signal_0\signals_O\signal_{i_1}) \cdot  \\
    &\qquad \qquad [\Pr(\signal_{i_1}\signal_{i_2}; \signals_{N}|\action_{i_2} = 1,\action_{i_1} = 1, \actions_{N}= \boldsymbol{0},\shirks_{N}, \signal_{i_1}) -\Pr(\signal_{i_1}\signal_{i_2}; \signals_{N}|\action_{i_2} = 0,\action_{i_1} = 1, \actions_{N}= \boldsymbol{0},\shirks_{N},\signal_{i_1})   ]
\end{align*}
with $s_{i_1}s_{i_2} \in \{AA, BB, AB, BA\}$ and ${\signal_0\signals_O} \in \{A,B\}^{|O|+1}$. Then, for each combination of $(\shirk_{i_2}, \shirks_{N})$, the IC becomes\begin{multline*}   
 \sum_{\signal_0\signals_O\signal_{i_1}}P^*(\signal_0\signals_O\signal_{i_1}) \cdot\\
 [\Pr(\signal_{i_1}\signal_{i_2}; \signals_{N}|\action_{i_2} = 1,\action_{i_1} = 1, \actions_{N}= \boldsymbol{0},\shirks_{N}, \signal_{i_1}) -\Pr(\signal_{i_1}\signal_{i_2}; \signals_{N}|\action_{i_2} = 0,\action_{i_1} = 1, \actions_{N}= \boldsymbol{0},\shirks_{N},\signal_{i_1})   ]w_{i_2}(\signals) \geq c.
\end{multline*}
By \textbf{Step 1.2}, for any fixed $w_{i_2}(\signals)$, we can define $\hat{w}_{i_2}(\signal_0\signals_{O}; \signal_{i_1},\signals_{N},\signal_{i_2})$ that is constant across $\signal_0\signals_{O} \in \{A,B\}^{|O|+1}$. Therefore, the left-hand side of IC simplifies into \begin{align*}
     &\sum_{\signal_{i_1}, \signal_{i_2}} [\Pr(\signal_{i_1}\signal_{i_2}; \signals_{N}|\action_{i_2} = 1,\action_{i_1} = 1, \actions_{N}= \boldsymbol{0},\shirks_{N}) -\Pr(\signal_{i_1}\signal_{i_2}; \signals_{N}|\action_{i_2} = 0,\action_{i_1} = 1, \actions_{N}= \boldsymbol{0},\shirks_{N})   ]\cdot\\
     &\qquad \qquad \hat{w}_{i_2}(\signal_{i_1}, \signal_{i_2}; \signals_{-(i_1, i_2)}) \quad \cdot \quad   \underbrace{\sum_{\signal_0\signals_O}P^*(\signal_0\signals_O\signal_{i_1}|\signal_{i_1})}_{=\Pr(\state_{i_1} = \signal_{i_1}) = \frac{1}{2}}.
\end{align*}
Then the set of IC exactly reduces to that in \textbf{Step 1.1}. Calculating using the same procedure, we show that in finding a lower bound of payment, the principal should set $w_{i_2}^*(\signal_{i_1} = \signal_{i_2}) = 2ce^{\lambda (t_{i_2}-t_{i_1})}> 0$ and $w_{i_2}^*(\signal_{i_1} \neq \signal_{i_2})  = 0$. The lower bound of payment for agent $i_2$ is then \begin{align*}
    \hat{K}_{i_2} &= \sum_{\signals \in \{A, B\}^{n+1}}w_{i_2}^*(\signals)P^*(\signals) = P^*(s_{i_1} = s_{i_2}) \cdot \contract_{i_2}^*(s_{i_1} = s_{i_2}) = c(1+e^{\lambda (t_{i_2} - t_{i_1})})
\end{align*}

Now, we turn to the sandwich case of \textbf{Step 1.3} ($0 \leq t_{i_2} \leq t_{i_1}$). Here, we work out the case with full generality (as in \textbf{Step 1.4}) by explicitly handling the set of non-neighbor risk-free agents, $M$ (which is empty  in the sandwich case of \textbf{Step 1.3}). According to the definitions in \textbf{Step 1.2}, the left-risk-free neighbor is the principal, i.e., $L = 0$, and the right-risk-free-neighbor is agent $i_1$, i.e., $R = i_1$. It proves helpful to divide the set of outer risky agents, $O$, into the \textit{left} outer risky agents, $O_l$, if they are located to the left of $L$;  and the \textit{right} outer risky agents, $O_r$, if they are located to the right of $R$. Similarly, define the set of left-risk-free agents (except for the neighbor) to be $M_l$, and the set for those on the right to be $M_r$.

First, we reduce agent $i_2$'s IC. For every fixed combination of $(\shirk_{i_2}, \shirks_N)$, the likelihood difference in the IC is 
\begin{align*}
    &[\Pr(\signals|\action_{i_2} = 1,\action_{i_1} = 1, \actions_{O} = \mathbf{1}, \actions_{N}= \boldsymbol{0},\shirks) - \Pr(\signals|\action_{i_2} = 0,\action_{i_1} = 1, \actions_{O} = \mathbf{1}, \actions_{N}= \boldsymbol{0},\shirks)  ] \\
    &=  P^*(\signals_{M_l},\signals_{O_l},\signal_{L}) \cdot  \\
    &\qquad [\Pr(\signal_{L}\signal_{i_2}\signal_{R}; \signals_{N}|\action_{i_2} = 1, \actions_{N}= \boldsymbol{0},\shirks_{N}, \signal_{L}, \signal_{R}) -\Pr(\signal_{L}\signal_{i_2}\signal_{R}; \signals_{N}|\action_{i_2} = 0, \actions_{N}= \boldsymbol{0},\shirks_{N},\signal_{L}, \signal_{R}) ]\cdot 
    \\
    &\qquad \qquad \qquad \qquad \qquad \qquad \qquad   P^*(\signal_R,\signals_{O_r},\signal_{M_r}) 
\end{align*}
By \textbf{Step 1.2}, for any fixed $w_{i_2}(\signals)$, we can define $\hat{w}_{i_2}(\signals_{M_l},\signals_{O_l}, \signals_{O_r},\signal_{M_r}; \signal_{L}, \signal_{R}, \signals_{N})$ that is constant across $(\signals_{M_l},\signals_{O_l}, \signals_{O_r},\signal_{M_r}) \in \{A,B\}^{|M| + |O|}$, given $(\signal_{L}, \signal_{R}, \signals_{N})$. Therefore, we can simplify the IC into 
\begin{align*}
&\sum_{\signal_{L}, \signal_{R}, \signal_{i_2}} [\Pr(\signal_{L}\signal_{i_2}\signal_{R}; \signals_{N}|\action_{i_2} = 1, \actions_{N}= \boldsymbol{0},\shirks_{N}) -\Pr(\signal_{L}\signal_{i_2}\signal_{R}; \signals_{N}|\action_{i_2} = 0, \actions_{N}= \boldsymbol{0},\shirks_{N}) ]
\hat{w}_{i_2}(\signal_{L}\signal_{i_2}\signal_{R}; \signals_{-(L,i_2,R)}) \ \cdot\\
&\qquad \qquad \underbrace{\sum_{\signals_{M_l},\signals_{O_l},\signals_{O_r},\signal_{M_r}} P^*(\signals_{M_l},\signals_{O_l}|\signal_{L}) \cdot  P^*(\signals_{O_r},\signal_{M_r}|\signal_R)}_{=\ \Pr(\state_{L} = \signal_{L}, \state_{R} = \signal_{R}) \ =\ \frac{1}{4}} \geq c 
\end{align*}
 for every $\shirk_{i_2}  \in \{A, B\}$ and every $\shirks_{N} \in \{A,B\}^{|N|}$. We can enumerate the likelihood difference \[
 [\Pr(\signal_{L}\signal_{i_2}\signal_{R}; \signals_{N}|\action_{i_2} = 1, \actions_{N}= \boldsymbol{0},\shirks_{N}) -\Pr(\signal_{L}\signal_{i_2}\signal_{R}; \signals_{N}|\action_{i_2} = 0, \actions_{N}= \boldsymbol{0},\shirks_{N}) ]
 \]
for each combination of $\signal_L\signal_{i_2}\signal_R$. When $\shirk_{i_2} = A$ and the inner risky agents' shirking signal profile is $\shirks_N$, the relevant likelihood differences are (we henceforth omit $\signals_N$ and $\actions_N = 0$):
\begin{align*}
       \Pr&(AAA|\action_{i_2} = 1,\shirk_{i_2} = A,\shirks_{N}) - \Pr(AAA|\action_{i_2} = 0,\shirk_{i_2} = A,\shirks_{N})\\
       &=\underbrace{\frac{1}{2}}_{\Pr(\signal_L = A)} \times 
       \underbrace{\frac{1}{2}(1+e^{-\lambda (t_{i_2}-t_L)})}_{\Pr(\state(t_L) = \state(t_{i_2}))} \times
       \underbrace{\frac{1}{2}(1+e^{-\lambda (t_{R} - t_{i_2})})}_{\Pr(\state(t_{i_2}) = \state(t_{R}))}  - \underbrace{\frac{1}{2}}_{\Pr(\signal_L = A)} \times \underbrace{1}_{\Pr(\signal_{i_2} = A|\shirk_{i_2} = A)} \times \underbrace{\frac{1}{2}(1+e^{-\lambda (t_{R} - t_L)})}_{\Pr(\state(t_L) = \state(t_{R}))}  
       = -\frac{1}{8} P_L^-P_R^-\\
       \Pr&(ABA|\action_{i_2} = 1,\shirk_{i_2} = A,\shirks_{N}) - \Pr(ABA|\action_{i_2} = 0,\shirk_{i_2} = A,\shirks_{N})\\
       &=\underbrace{\frac{1}{2}}_{\Pr(\signal_L = A)} \times 
       \underbrace{\frac{1}{2}(1-e^{-\lambda (t_{i_2}-t_L)})}_{\Pr(\state(t_L) \neq  \state(t_{i_2}))} \times
       \underbrace{\frac{1}{2}(1-e^{-\lambda (t_{R} - t_{i_2})})}_{\Pr(\state(t_{i_2}) \neq \state(t_{R}))}  - \underbrace{\frac{1}{2}}_{\Pr(\signal_L = A)} \times \underbrace{0}_{\Pr(\signal_{i_2} = B|\shirk_{i_2} = A)} \times \underbrace{\frac{1}{2}(1+e^{-\lambda (t_{R}-t_L)})}_{\Pr(\state(t_L) = \state(t_{R}))}  
       = \frac{1}{8} P_L^-P_R^-\\
        \Pr&(ABB|\action_{i_2} = 1,\shirk_{i_2} = A,\shirks_{N}) - \Pr(ABB|\action_{i_2} = 0,\shirk_{i_2} = A,\shirks_{N}) = \frac{1}{8} P_L^-P_R^+\\
       \Pr&(AAB|\action_{i_2} = 1,\shirk_{i_2} = A,\shirks_{N}) - \Pr(AAB|\action_{i_2} = 0,\shirk_{i_2} = A,\shirks_{N}) = -\frac{1}{8} P_L^-P_R^+\\
       \Pr&(BBA|\action_{i_2} = 1,\shirk_{i_2} = A,\shirks_{N}) - \Pr(BBA|\action_{i_2} = 0,\shirk_{i_2} = A,\shirks_{N}) = \frac{1}{8} P_L^+P_R^-\\
       \Pr&(BAA|\action_{i_2} = 1,\shirk_{i_2} = A,\shirks_{N}) - \Pr(BAA|\action_{i_2} = 0,\shirk_{i_2} = A,\shirks_{N}) = -\frac{1}{8} P_L^+P_R^-\\
       \Pr&(BAB|\action_{i_2} = 1,\shirk_{i_2} = A,\shirks_{N}) - \Pr(BAB|\action_{i_2} = 0,\shirk_{i_2} = A,\shirks_{N}) = -\frac{1}{8} P_L^+P_R^+\\
       \Pr&(BBB|\action_{i_2} = 1,\shirk_{i_2} = A,\shirks_{N}) - \Pr(BBB|\action_{i_2} = 0,\shirk_{i_2} = A,\shirks_{N}) = \frac{1}{8} P_L^+P_R^+
   \end{align*}
   where $ P_L^+ = 1 + e^{-\lambda (t_{i_2}-t_L)},  P_L^- = 1 - e^{-\lambda (t_{i_2}-t_L)},  P_R^+ = 1 + e^{-\lambda (t_R - t_{i_2})},  P_R^- = 1 - e^{-\lambda (t_R - t_{i_2})}$. Plugging the above into the constraint with $\shirk_{i_2} = A$, and substituting $t_L = 0$, $t_R = t_{i_1}$, 
   we get \eqref{ICA1.3}.
Similar calculations yield \eqref{ICB1.3}. 
Finally, we illustrate how to verify
\[
\frac{\Pr(AAA|\mathbf{a} = \mathbf{1})}{\frac{1}{8}[(P_L^+P_R^+)^2 - (P_L^-P_R^-)^2]} \leq \min\left\{\frac{\Pr(AAB|\mathbf{a} = \mathbf{1})}{\frac{1}{8}[P_R^+P_R^- ((P_L^+)^2 - (P_L^-)^2)]}, \frac{\Pr(BAA|\mathbf{a} = \mathbf{1})}{\frac{1}{8}[P_L^+P_L^- ((P_R^+)^2 - (P_R^-)^2)]}\right\}
\]
For the first term on the right-hand side, the inequality becomes
$\frac{\frac{1}{8}P_L^+P_R^+}{\frac{1}{2}(1+e^{-\lambda t_{i_1}})(e^{-\lambda (t_{i_1} - t_{i_2})} + e^{-\lambda t_{i_2}})} \leq \frac{\frac{1}{8}P_L^+P_R^-}{\frac{1}{2}P_R^+P_R^-e^{-\lambda t_{i_2}}}$, which simplifies to $2e^{-\lambda t_{i_1}} \leq e^{-\lambda (t_{i_1} - t_{i_2})} + e^{-\lambda (t_{i_1}+t_{i_2})}$. It holds (strictly) as long as $t_{i_1} \geq (>) t_{i_2}$. The other inequality is obtained by similar calculations.

}

\end{document}